\documentclass[runningheads]{llncs}
\usepackage[dvipsnames]{xcolor}
\usepackage[]{amssymb}
\usepackage{url}
\usepackage{graphicx}
\usepackage{subcaption}
\usepackage{enumitem}
\usepackage{mathtools}
\usepackage{xcolor}
\usepackage{fancybox}

\newcommand{\ee}{\tt{ee}}

\usepackage{tikz}
\usetikzlibrary{arrows.meta,shapes}

\tikzset{
     punt/.style={draw,circle,fill,inner sep=0pt,minimum size=1mm},
     event/.style={draw,fill,circle,inner sep=0pt,minimum size=1.5mm}
     }

\tikzset{
      time/.style={-Stealth},
      message/.style={-Latex,densely dashed},
      alphamap/.style={-Stealth},
      mexmap/.style={dashed,Latex-Latex,thick},
      thetamap/.style={dash pattern=on 2pt off 2pt,Latex-,thick}, 
      mindmap/.style={dotted,->},
      test/.style={-{Straight Barb[angle'=70]},thick}
      }
\tikzset{
    envelope/.pic={\draw (-2.5,0) rectangle (2.5,3) 
      (-2.5,3) -- (0,1) -- (2.5,3) 
      (-2.5,0) -- (-0.5,1.5) (0.5,1.5) -- (2.5,0) ; },      
     }

\titlerunning{Cross-organisational Process Mining from Message Logs}

\author{Pieter Kwantes, Jetty Kleijn }

\newcommand{\cf}            {cf. }
\newcommand{\ie}         	{i.e., }
\newcommand{\eg}			{e.g., }

\newcommand{\ctd}         	{cont.}

\newcommand{\Ev}			{{\tt EL}}

\newcommand{\Ln}			{E}
\newcommand{\LN}			{\mathcal{E}}

\newcommand{\Gn}			{I}
\newcommand{\GN}			{\mathcal{I}}

\newcommand{\V}			{V}

\newcommand{\DA}		{\text{\emph{DCA}}}

\newcommand{\DAx}		{\ensuremath{\mathcal{D\!A}}}

\newcommand{\pref}			{{\tt pref}}
\newcommand{\Pref}			{{\tt Pref}}

\newcommand{\dom}			{{\tt dom}}
\newcommand{\card}			{{\tt card}}

\newcommand{\pos}			{{\tt pos}}

\newcommand{\occ}			{{\tt occ}}
\newcommand{\tp}			{{\tt type}}

\newcommand{\dc}			{{\tt past}}

\newcommand{\proj}			{{\tt proj}}

\newcommand{\ab}				{{\tt alph}}
\newcommand{\Ab}				{{\tt Alph}}

\newcommand{\lin}				{{\tt lin}}

\newcommand{\fifo}			{{\tt fifo}}

\newcommand{\Att}			{{\tt Att}}

\newcommand{\out}			{{\tt out}}
\newcommand{\inp}			{{\tt inp}}

\newcommand{\lab}			{{\tt label}}

\newcommand{\Ka}			    {\ensuremath{\mathbb{K}}}
\newcommand{\K}			    {\ensuremath{\Lambda}}

\newcommand{\Cs}			{{\tt case}}
\newcommand{\msg}			{{\tt msg}}

\newcommand{\m}		{{\tt m}}

\newcommand{\I}			    {\ensuremath{\Upsilon}}

\newcommand{\C}			    {\ensuremath{\mathcal{C}}}

\newcommand{\x}			{{\tt x}}
\newcommand{\y}			{{\tt y}}

\newcommand{\PR}			    {\ensuremath{\mathcal{P}}}

\newcommand{\proc}			{{\tt proc}}

\newcommand{\tm}			{{\tt time}}

\newcommand{\Ch}		{\ensuremath{{\tt sr}}}

\newcommand{\wdspo}			{\ensuremath{{\tt seq}}}
\newcommand{\po}			{\ensuremath{{\tt po}}}

\newcommand{\lo}			{\ensuremath{{\tt lo}}}

\newcommand{\lex}			{\ensuremath{{\tt linext}}}

\newcommand{\llim}			{\ensuremath{{\tt lim}}}
\newcommand{\cp}			{\ensuremath{{\tt mex}}}
\newcommand{\Id}			{\ensuremath{{\tt Id}}}

\definecolor{light-gray}{gray}{0.95}
\newcommand{\LK}   {\ensuremath{\boldsymbol{L}}}

\newcommand{\Inv}         	{{\tt IF}}
\newcommand{\Cus}         	{{\tt CN}}

 \spnewtheorem*{definition4}{Definition~4}{\bfseries}{\itshape}
  \spnewtheorem*{definition6}{Definition~6}{\bfseries}{\itshape}
 \spnewtheorem*{definition8}{Definition~8}{\bfseries}{\itshape}
  \spnewtheorem*{definition9}{Definition~9}{\bfseries}{\itshape}
 \spnewtheorem*{definition10}{Definition~10}{\bfseries}{\itshape}
 \spnewtheorem*{definition11}{Definition~11}{\bfseries}{\itshape}
  \spnewtheorem*{definition12}{Definition~12}{\bfseries}{\itshape}
    \spnewtheorem*{definition13}{Definition~13}{\bfseries}{\itshape}
 \spnewtheorem*{theorem3}{Theorem~3}{\bfseries}{\itshape}
 \spnewtheorem*{theorem99}{Theorem~13*}{\bfseries}{\itshape}
  \spnewtheorem*{corollary4}{Corollary~4}{\bfseries}{\itshape}

\begin{document}

\title{Cross-organisational Process Mining 
\\
from Message Logs
}
\author{
Pieter Kwantes
\and
Jetty Kleijn
}
\institute{
              LIACS,
              Leiden University,
              P.O.Box 9512\\
              NL-2300 RA Leiden, The Netherlands\\
\email{\{p.m.kwantes,h.c.m.kleijn\}@liacs.leidenuniv.nl}
}

\maketitle

\begin{abstract} 
Enterprise nets and Industry nets are specific Petri net models that together form a framework for the modelling of collaborating business processes and the specification of interactions between them. 
In this set-up, individual business process are modelled as Enterprise nets, while Industry nets consist of Enterprise nets that interact by exchanging messages through unidirectional channels. 
In~\cite{DBLP:conf/apn/KwantesK20,DBLP:journals/topnoc/KwantesK22}, it was 
established how to leverage algorithms for the discovery of single processes (modelled as Enterprise nets) from event logs (consisting of finite sequences of executed actions), 
to synthesise systems consisting of asynchronously communicating  processes (modelled as Industry nets). 
However, in general, it may be difficult to work with event logs from different organisations or even get access to them. 
Recordings of message exchanges between processes may be easier available. 
Therefore, we focus in this paper on the communications between different  processes and turn to the investigation of message logs, \ie collections of recordings of communications in the form of message exchanges (output and input events) between two processes.
In particular, the question is addressed of how a message log can be transformed into a description of the communication behaviour in the form of an event log such that the causal order on message events is preserved. 
Following the approach from~\cite{DBLP:conf/apn/KwantesK20,DBLP:journals/topnoc/KwantesK22}, this event log can then be used to synthesis an Industry net.\footnote{This paper is the full version of an extended abstract with the same title, presented at the third International workshop on Collaboration Mining for Distributed Systems (COMINDS 2024) held in conjunction with ICPM 2024 in Copenhagen.}

\end{abstract}

%%%%%%%%%%%%%%%%%%%%%%%%%%%%%%%%%%%%%%%%%%%%%%%%%%%%%
%%%%%%%%%%%%%%%%%%%%%%%%%%%%%%%%%%%%%%%%%%%%%%%%%%%%%

\section{Introduction}
\label{Sect-Intro}

%%%%%%%%%%%%%%%%%%%%%%%%%%%%%%%%%%%%%%%%%%%%%%%%%%%%%

Process mining is concerned with the algorithmic discovery of models of business processes from event data, to check their conformance, analyse bottlenecks, and suggest improvements, (see, \eg\cite{DBLP:books/sp/Aalst16,DBLP:books/sp/22/PMH2022,DBLP:journals/tkde/AugustoCDRMMMS19,DBLP:conf/birthday/2026aalst}). 
So far, the focus of process mining research has been on discovering \emph{local} business processes, \ie residing within the boundaries of a single organisation or enterprise.
The interaction between business processes from different organisations has received less attention (see, \eg\cite{DBLP:journals/ejis/MendlingPR20,Rott2024,Thiede2018}).
Nevertheless, in many cases, organisations collaborate with other organisations, thereby forming \emph{cross-organisational processes}, \ie conglomerates of interacting local business processes. 
The behaviour of a cross-organisational process is a combination of the local behaviours of its individual business processes subject to their interactions. 
On the other hand, the behaviour of a local business process within a cross-organisational process will in general be affected by interactions with the other local business processes. 
The extension of established process mining methods and techniques from local processes of individual organisations to \emph{global} (at the level of the collaborating organisations) processes appears in the research agenda in the Process Mining Manifesto~\cite{DBLP:conf/bpm/AalstAM11} as an open challenge. 

Petri nets are a well-established formal framework for the modelling of a wide range of distributed systems 
(see, \eg\cite{DBLP:journals/topnoc/2013-7,AdvCourse2023,DBLP:conf/ac/1996petri1,DBLP:conf/ac/1996petri2}). 
Also the modelling of business processes is traditionally an important application domain of Petri nets, see~\cite{DBLP:journals/jcsc/Aalst98,DBLP:books/daglib/0032051,DBLP:conf/birthday/2026aalst}. 
A specific class of Petri nets that supports a formal algorithmic analysis of business processes and that is often used in practice are the \emph{workflow nets}, 
originally introduced in~\cite{DBLP:conf/apn/Aalst97}. 

In~\cite{DBLP:conf/apn/KwantesK18}, a framework was proposed to specify the composition of a comprehensive global process model (an \emph{Industry net}) from a set of local process models (\emph{Enterprise nets}).  
Enterprise nets are Petri nets with typed input and output transitions, and internal transitions. 
They can be considered as a generalisation of workflow nets.
Industry nets consist of separate Enterprise nets that interact by exchanging messages through unidirectional channels (in the form of additional places connecting a single output transition of one Enterprise net with a single input transition with the same type, of another Enterprise net). 
A comprehensive overview with a discussion of this and other models can be found in~\cite{DBLP:conf/icpm/BenzinR23}. 
As argued in~\cite{DBLP:conf/apn/KwantesK18}, the set-up proposed there ensures that global compliance of an Industry net with a reference model can be verified by local checks (of Enterprise nets) only. 
As stated in~\cite{DBLP:conf/bpm/AalstAM11}, it may be difficult to work with local information from different organisations as these are often reluctant to disclose their logs to the public domain. 
Furthermore, even if local information is available, there is the problem of so-called event correlation (see also~\cite{DBLP:conf/bpm/AalstAM11}) to determine how local sequences of actions combine into global behaviour.  

Next, in the workshop paper \cite{DBLP:conf/apn/KwantesK20} and its extended version~\cite{DBLP:journals/topnoc/KwantesK22}, focus is on synthesising Industry nets from the observed global (combined) behaviour of collaborating Enterprise nets. 
More specifically, it is investigated to what extent existing algorithms for the discovery of local process models (Enterprise nets) from \emph{event logs}, \ie sets of finite sequences of events (occurrences of actions) representing the behaviour of a process, can be lifted to the synthesis of systems consisting of asynchronously communicating local processes (modelled as Industry nets), from global event logs.
This requires a description of the distribution of actions over the components together with their role (internal, input, or output) in the form of a so-called \emph{distributed communicating alphabet}. 

Whereas the information in local event logs is confined to single enterprises, communication involves multiple enterprises. 
Therefore, to address the problems mentioned earlier, we focus in this paper on the communications between local processes and turn to the investigation of \emph{message logs}, collections of recordings of communications in the form of message exchanges between two processes, typically of typed messages over bilateral unidirectional channels.
The \emph{type} of a message represents its format (\eg prescribed by the prevailing industry standard\footnote{Examples include the ISO15022 and ISO20022 industry standard for financial services, supported by SWIFT, RosettaNet for the electronics industry, and EDSN for the Dutch energy market, to name just a few.} for groups of messages with similar processing characteristics). 
The subject of a message is referred to as the \emph{case} of the message (\eg an order number to identify a specific order). 
A single execution (or run) of a cross-organisational process comprises all events executed to complete a particular case. 
Recordings of such message exchanges appear as byproducts of day-to-day operations of message service providers (see, \eg\cite{DBLP:journals/isem/EngelKZPBAWH16} for case studies exploring the use of recorded messages for process mining). 
Here, message exchanges will be defined in the form of two \emph{message events}, an output event (representing the sending of the message by a local process) and an input event (representing the receipt of that message by a second local process). 

In this paper, we address the question of how message logs can be transformed into global event logs in such a way that essential properties of the behaviour of the global system as captured by the message log, are preserved. 
Following the approach from~\cite{DBLP:conf/apn/KwantesK20,DBLP:journals/topnoc/KwantesK22}, 
message logs thus transformed, can then be used for the discovery of an Industry net.

The paper is organised as follows. 
Section~\ref{Sect-Prelim} provides notation and notions used throughout this paper, including concepts relating to sequences (words), sets of sequences (languages), and partial orders.
Special attention is given to the sequentialisation of partial orders. 
We argue how infinite partial orders can be represented by sets of finite sequences of actions.
These observations are crucial in the following sections. 
In Section~\ref{Sect-MessLog}, message events and message logs are introduced. 
A message log is a (possibly infinite) set of message events together with their \emph{attributes}, \ie parameters giving the case to which an event belongs and its type, as well as the \emph{process} executing the event and a local \emph{time}.  
More precisely, each process has a time function which assigns time stamps to its message events in such a way that no two different events executed by the same process are assigned the same local time.
Hence, the events of a local process are always totally ordered by their time stamps.
This total order is assumed to be consistent with the causal dependencies between message events in the sense that whenever a message event $e$ has an earlier time stamp than another message event $e'$ executed by the same local process, then $e'$ does not causally precede $e$ (\ie it is not necessary to execute $e'$ before $e$).
Next, we discuss in Section~\ref{Sect-MLandPO}, a natural precedence relation on message events belonging to the same case, and how and when this defines a partial order.  
As discussed in Section~\ref{Sect-Prelim}, this (possibly infinite) partial order can then be used to define a set of finite sequences of message events, that together describe one of the runs of the system captured by the message log. 
In Section~\ref{Sect-DCA}, the aim is to associate an event log and a distributed communicating alphabet with a given message log. 
This entails an interpretation of output events and input events as occurrences of output actions and input actions, respectively, leading to a labelling of message events by action names.
We investigate how to consistently identify message events in a message log that are to be labelled with the same action name. 
This leads to the concept of an \emph{action function} that can be used to translate the set of finite sequences of message events defined by a message log into a global event log consisting of sequences of labelled message events. 
Moreover, action functions can be used to define a suitable distributed communicating alphabet which on the one hand describes the distribution of output and input actions over components, and on the other hand how they are paired up into unidirectional bilateral channels used to exchange the messages recorded in the message log. 
Clearly, as the global event logs thus derived are defined by identifying message events, they are more abstract than the original message logs. 
So, the main question addressed in Section~\ref{Sect-ML} is whether the causalities in a message log as expressed by the partial order derived from the precedence relation, are faithfully preserved in the global event log defined by an action function.   
To settle this question we turn to the assignment functions introduced in~\cite{DBLP:journals/topnoc/KwantesK22}, to relate occurrences of input and output actions in the labelled sequences constituting a global event log. 
It is demonstrated how the causal partial orders underlying the sequences of the message log coincide with the partial orders on the occurrences of actions defined by message preserving assignment functions.  
Moreover, all sequences defined by the partial orders of message preserving assignments are in the global event log. 
Thus, given a message log and an action function, a distributed alphabet and a global event log can be derived that faithfully represents the causalities in the message log. 
We also discuss assignment functions that are oblivous to the identities of the message exchanged and model a first-in-first-out channel policy based on output and input actions.  
In Section~\ref{Sect-Concl}, we combine all technical results obtained and apply the  procedure from~\cite{DBLP:journals/topnoc/KwantesK22} to synthesise on the basis of a given process discovery algorithm an Industry net that generates at least all sequences from the event log. 
Finally, Section~\ref{Sect-Disc} reflects on the research presented.
In the Appendix, the proof of Lemma~\ref{Lem-seqlimpref} is given. 
This auxiliary result is crucial for the representation of infinte partial orders by finite sequences. 
The Appendix also provides the definitions for Petri nets, Enterprise nets and Industry nets, as well as definitions of key concepts and results from~\cite{DBLP:journals/topnoc/KwantesK22}.

%%%%%%%%%%%%%%%%%%%%%%%%%%

\subsection*{Related work}
\label{Sect-RelWork}

%%%%%%%%%%%%%%%%%%%%%%%%%%

Using the messages exchanged between collaborating enterprises as input for process mining and process discovery, has been proposed before. 
In particular, the use of messages exchanged between web services (from different enterprises) has been described, \eg in \cite{DBLP:journals/debu/AalstV08,AalstCCSB2008,DBLP:conf/caise/EngelAZPW12,DBLP:conf/ecweb/EngelKZPAW11,DBLP:journals/isem/EngelKZPBAWH16}. 
In \cite{DBLP:journals/debu/AalstV08}, information from the so-called web services stack is used to obtain event logs.  
According to~\cite{AalstCCSB2008}, global message logs comprising the messages from multiple parties involved in a collaboration of web services, can be reconstructed with middleware technology and translated into an event log including a case id, an activity name and a time stamp.
It is then shown how this event log can be used to check conformance of the event log with a Petri net model of the business process, derived from a BPEL specification.
Several case studies (\cite{DBLP:conf/caise/EngelAZPW12,DBLP:conf/ecweb/EngelKZPAW11,DBLP:journals/isem/EngelKZPBAWH16}) describe how a message log, gathered at a particular enterprise, can be used to derive an event log and subsequently, using existing process discovering algorithms, a process model. 
For example,~\cite{DBLP:journals/isem/EngelKZPBAWH16} proposes an approach for cross-organisational process mining, based on the use of EDI-messages, which is implemented in the EDI miner toolset~\cite{DBLP:conf/caise/EngelBPZW13}. 

Based on a review of 66 articles, the paper~\cite{Rott2024} provides an  overview of current research. 
It concludes that research into the field has been dispersed and gives an agenda for future research with a focus on socio-technical aspects is proposed. 
In~\cite{DBLP:conf/icpm/BenzinR23,DBLP:conf/icpm/BenzinR24} a survey of different approaches to modelling cross-organisational processes is provided. 
This leads to the conclusion that a standard modelling framework is lacking and to an investigation of a more standardised approach based on Petri nets.

%%%%%%%%%%%%%%%%%%%%%%%%%%%%%%%%%%%%%%%%%%%%%%%%%%%%%
%%%%%%%%%%%%%%%%%%%%%%%%%%%%%%%%%%%%%%%%%%%%%%%%%%%%%

\section{Preliminaries}
\label{Sect-Prelim}

%%%%%%%%%%%%%%%%%%%%%%%%%%%%%%%%%%%%%%%%%%%%%%%%%%%%%

%%%%%%%%%%%%%%%%%%%%%%%%%%%%%

\subsection*{General notions}

%%%%%%%%%%%%%%%%%%%%%%%%%%%%%

By $\mathbb{N}$ we denote the set of natural numbers $\{0,1,2, \ldots \}$. 
For $n \in \mathbb{N}$, $[n]=\{1,2, \ldots , n\}$. 
Hence, $[0]=\emptyset$. 

Let $A$ be a set.
If $A$ is a finite set, then $\card(A)$ denotes its cardinality. 
Let ${\mathbb A} = \{A_i \mid A_i \subseteq A, i \in {\mathcal I}\}$ with ${\mathcal I}$ an index set, be such that the $A_i$ are pairwise disjoint, non-empty sets, and $A = \bigcup\{A_i \mid i \in {\mathcal I}\}$. 
Then ${\mathbb A}$ is a \emph{partition} of $A$.

Functions are total unless explicitly stated otherwise. 
For a function $f: A \rightarrow B$ and $C \subseteq A$, we define $f(C)=\{f(c) \mid c \in C \}$ and the \emph{restriction} of $f$ to $C$ is the function $f|_C: C \rightarrow B$, defined by $f|_C(c)=f(c)$ for all $c \in C$. 
A function $f: A \rightarrow A$ 
is a \emph{complement function} if 
$A = B \cup C$ with $B$ and $C$ disjoint, 
$f(B) \subseteq C$, $f(C) \subseteq B$, and $f(f(a)) = a$ for all $a \in A$. 
Note that such $f$ is injective, because whenever $f(a) = f(a')$ for some $a,a' \in A$, then $a = f(f(a)) = f(f(a')) = a'$. 
Moreover, $f(B) = C$, since $c = f(f(c))$ for all $c \in C$. 
Similarly, $f(C) = B$.
Thus, $f$ is a bijection.

%%%%%%%%%%%%%%%%%%%%%%%%%%%%%%%%%

\subsection*{Words and Languages}

%%%%%%%%%%%%%%%%%%%%%%%%%%%%%%%%%

Throughout the paper, we will be dealing with sequences of events, actions, labels, etc., represented as symbols. 
Sets of symbols may be empty, finite, or infinite. 
An \emph{alphabet} is a set of symbols which is finite and non-empty, unless explicitly specified otherwise. 

Let $\Sigma$ be a set of symbols. 
A \emph{word} (over $\Sigma$) is a possibly infinite sequence of symbols (from $\Sigma$).
As usual, $\Sigma^*$ is the set of all finite words over $\Sigma$, $\Sigma^{\omega}$ the set of all infinite words (or $\omega$-words) over $\Sigma$, and $\Sigma^{\infty}=\Sigma^* \cup \Sigma^\omega$ the set of all words over $\Sigma$.
Any subset $L$ of $\Sigma^{\infty}$ is a \emph{language} (over $\Sigma$). 
If $L \subseteq \Sigma^*$, then $L$ is \emph{finitary} and if $L \subseteq \Sigma^\omega$, then it is an \emph{infinitary} language.
For a finite word $w=a_1 \cdots a_n$ with $n \ge 0$ and $a_i \in \Sigma$ for all $i \in [n]$, we refer to $n$ as the \emph{length} of $w$, denoted by $|w|$. 
If $n = 0$, then $w = \lambda$, the \emph{empty word} (over $\Sigma$). 
It is often convenient to consider words as functions. 
Thus, a finite word $w=a_1 \cdots a_n$ with
$a_i\in\Sigma$ for all $i \in [n]$, is defined by 
$w: [n] \rightarrow \Sigma$ with $w(i)=a_i$ for all $i \in [n]$;  
and an infinite word $v=a_1a_2 \cdots$ is $v: \mathbb{N} \rightarrow \Sigma$ with $v(i)=a_i$ for all $i \in \mathbb{N}$. 
The domain of a word $u\in\Sigma^\infty$, denoted $\dom(u)$, is the set of integers on which $u$ is defined.
Hence we have $\dom(u) = [n]$ if $|u| = n$ and $\dom(u) = {\mathbb N}$ if $u$ is infinite.
The \emph{alphabet} of $u$, denoted by $\ab(u)$, is the set of symbols that occur in $u$. %
Thus $\ab(u)= \{a \in \Sigma \mid \exists i \in \dom(u): u(i)=a\}$ and, in particular, $\ab(\lambda)=\emptyset$. 
Moreover, the alphabet of $u$ is infinite whenever $u$ is an infinite word in which an infinite number of different symbols occur (hence in this case $\Sigma$ must be an infinite set of symbols). 
For a language $L$, 
$\Ab(L) = \{a \in \Sigma \mid \exists w \in L: a \in \ab(w)\}$ 

is the set of all symbols that occur in some word of $L$. 
\\
Let $w\in \Sigma^*$ and $a \in \Sigma$. 
Then $\#_a(w)= \card(\{i \mid w(i)=a\})$ is the number of occurrences of $a$ in $w$ 
and $\occ(w)=\{(a,i) \mid a \in \ab(w) \wedge 1 \le i \le \#_a(w) \}$ is the set of all \emph{occurrences} of symbols from $\Sigma$ in $w$. 
The function $\pos_w : \occ(w) \rightarrow [|w|]$ gives for each $(a,i) \in \occ(w)$, its position in $w$: 
$\pos_w((a,i))=k$ if $w(k)=a$ and $\#_a (w(1) \cdots w(k))=i$. 
\\
Any function $f: \Sigma \rightarrow \Delta$ where $\Delta$ is a set of symbols can be extended to a function $f^\infty: \Sigma^\infty \rightarrow \Delta^\infty$ by setting $f^\infty (\lambda) = \lambda$ and $(f^\infty (w))(i) = f(w(i))$ for all $w \in \Sigma^\infty$ and $i\in\dom(w)$. 
Omitting the superscript $\infty$, we usually write $f$ instead of $f^\infty$. 
\\
Let now $\Delta \subseteq \Sigma$. 
Then the  \emph{projection} of $\Sigma$ on $\Delta$ is the function 
$\proj_{\Sigma,\Delta}: \Sigma^* \rightarrow \Delta^*$,
defined by
$\proj_{\Sigma,\Delta}(a)=a$ if $a \in \Delta$; 
$\proj_{\Sigma,\Delta}(a)=\lambda$ if $a \in (\Sigma\setminus\Delta) \cup \{\lambda\}$; 
and 
for all $a \in \Sigma$ and $w \in \Sigma^*$, $\proj_{\Sigma,\Delta}(aw)=\proj_{\Sigma,\Delta}(a)\proj_{\Sigma,\Delta}(w)$. 
If $\Sigma$  is clear from the context, we may omit the reference to $\Sigma$ and write $\proj_{\Delta}$. 

The \emph{(left-)concatenation} of $u \in \Sigma^*$ and $v \in \Sigma^\infty$, denoted as $uv$, is the word $w \in \Sigma^\infty$ defined by $w(i)=u(i)$ for all $i \in \dom(u)$ and  $w(|u|+ i)=v(i)$ for all $i \in \dom(v)$. 
Let $w \in \Sigma^\infty$. 
If $u \in \Sigma^*$, is such that $w = uv$ for some $v \in \Sigma^\infty$, then $u$ is a \emph{prefix} of $w$. 
For $n\in \dom(w)$, we denote by $w[n]$ the prefix of $w$ of length $n$. 
The set of all prefixes of $w$ is denoted by $\pref(w)$. 
Note that we only deal with finite prefixes and so $\pref(w) \subseteq \Sigma^*$.  
For a language $L \subseteq \Sigma^\infty$, its set of prefixes is $\Pref(L) = \{u \in \Sigma^* \mid \exists w \in L: u \in \pref(w)\}$. 
We say that $L$ is \emph{prefix-closed} if $\Pref(L) \subseteq L$. 
It should be noted here that, $L \subseteq \Pref(L)$ only holds for finitary languages. 
On the other hand, any infinite word can be viewed as the limit of its prefixes. 
More generally, let $L \subseteq \Sigma^*$. 
Then the \emph{limit} of $L$ 
(see, \eg \cite{DBLP:journals/tcs/EngelfrietH93,DBLP:reference/hfl/Staiger97}) 
is the language $\llim(L)=\{w \in \Sigma^{\omega} \mid w[i] \in L \mbox{ for infinitely many } i \in \mathbb{N} \}$. 
Clearly, $\llim(\pref(w)) = \{w\}$ for every infinite word $w$. 
As a consequence, $L \subseteq \llim(\Pref(L))$ holds for every infinitary language $L$. 
Such inclusion may be strict, as there are infinitary languages $L$ such that 
$\llim(\Pref(L))\not\subseteq L$.  

\begin{example}\label{Ex-lim=more}
    Consider $L = \{a^kba^\omega \mid k \in {\mathbb N}\}$. 
    Then $\Pref(L) = \{a^kba^\ell \mid k,\ell \in {\mathbb N}\}\cup \{a\}^*$. 
    And so, $\llim(\Pref(L)) = L \cup \{a^\omega\}$.
\qed
\end{example}

Finally, we recall a property of words (and languages) from~\cite{DBLP:journals/topnoc/KwantesK22}, introduced there to formalise how in an action sequence, every occurrence of an input action is preceded by a corresponding occurrence of an output action.
Let $\Sigma_o,\Sigma_\iota \subseteq \Sigma$ 
with $\Sigma_o \cap \Sigma_\iota = \emptyset$ and let 
$f: \Sigma_o \rightarrow \Sigma_\iota$ be a bijection. 
Let $w \in \Sigma^\infty$. Then $w$ has the \emph{prefix property with respect to $f$} if, 
for all $v\in\pref(w)$ and for all $a \in \Sigma_o$, 
$\#_{a} (v) \geq \#_{f(a)} (v)$.
A language $L \subseteq \Sigma^\infty$ is said to have the prefix property with respect to $f$ if all $w \in L$ have the prefix property with respect to $f$.

%%%%%%%%%%%%%%%%%%%%%%%%%%%%%%%%%%%%%%%%%%%%%

\subsection*{Partial orders and total orders} 

%%%%%%%%%%%%%%%%%%%%%%%%%%%%%%%%%%%%%%%%%%%%

Let $A$ be a set. 
Let $R \subseteq A \times A$ be a binary relation on $A$.
The \emph{identity relation} on $A$ is $\Id_A = \{(a,a) \mid a \in A \}$. 
In what follows, we will mostly use the infix notation for binary relations. 
Thus we write $a\, R \, b$ for $(a,b) \in R$. 
Relation $R$ is \emph{reflexive} if $\Id_A \subseteq R$, \ie if $a\,R\,a$ for all $a \in A$; 
it is \emph{irreflexive} if there is no $a \in A$ such that $a\,R\,a$; 
it is \emph{transitive} if whenever $a\,R\,b$ and $b\,R\,c$, then also $a\,R\,c$; 
it is \emph{antisymmetric} if whenever both $a\,R\,b$ and $b\,R\,a$, then $a = b$. 
The \emph{transitive closure} $R^+$ of $R$ is defined by $R^+ = \{(a,b) \mid \exists n \geq 1,  \text{ } a_0, a_1, \ldots , a_n \in A \text{ such that } a_0 = a, a_n = b \text{, and  $\forall i\in [n]$:  } a_{i-1} \, R \, a_i \}$. 
The \emph{reflexive, transitive closure} of $R$ is $R^* = R^+ \cup \Id_A$. 
$R$ is said to be \emph{acyclic} if $R^+ \setminus \Id_A$ is irreflexive, \ie there 
%are no $a,b \in A$ such that $a \neq b$ and $(a,b),(b,a) \in R^+$. 
are no $a, a_1, \ldots , a_n = a \in A$, $n \geq 1$ with $(a,a_1), \ldots , (a_{n-1},a) \in  R$.  
\\
If $R$ is reflexive, transitive, and antisymmetric, it is a \emph{partial order} on $A$. 
If $R$ is a partial order such that $x R y$ or $y R x$, for all $x,y \in A$, then $R$ is a \emph{linear order} (or \emph{total order}). 
A \emph{partially ordered set}, or \emph{poset}, is a pair $(A,\leq_A)$, where $A$ is a set and $\leq_A$ is a partial order on $A$. 
If $\leq_A$ is a linear order, we refer to $(A,\leq_A)$ as a linearly ordered set. 

Let $\po=(A,\leq_A)$ be a poset. 
Poset $(B,\leq_B)$ is said to be a \emph{sub-poset} of $\po$ if $B \subseteq A$ and $\leq_B \;\subseteq\; \leq_A$. 
If $B \subseteq A$ and $\leq_B\; =\;\leq_A \; \cap \; (B \times B)$, then we refer to $(B,\leq_B)$ as the sub-poset \emph{induced} by $B$. 
A \emph{linear extension} of the poset $\po=(A,\leq_A)$ is a poset $\lo=(A,\leq')$ such that $\leq'$ is a linear order on $A$ and $\leq_A \;\subseteq\; \leq'$. 
In such case, we also loosely call $\leq'$ a linear extension of $\leq$. 
Note that every poset is a sub-poset of each of its linear extensions. 
We use $\lex(\po)$ to denote the set of all linear extensions of $\po$.  
The linear extensions of a poset are maximal in the sense that adding a new relation between two distinct elements leads to a cycle (and so the result is no longer a partial order). 
By Szpilrajn's theorem~\cite{Szpilrajn} every poset has a linear extension. 
Moreover, as a consequence, every partial order is the intersection of its linear extensions. 
\\
If every non-empty subset $B$ of $A$ has a minimal element with respect to $\leq_A$ (\ie there exists an element $b \in B$ such that for all $a \in B$ whenever $a \leq_A b$ then $a = b$), then both the partial order $\leq_A$ and the poset $(A,\leq_A)$ are said to be \emph{well-founded}\footnote{An equivalent definition is: $\leq_A$ is well-founded iff there is no infinite sequence $a_1, a_2, \ldots $ of distinct $a_i \in A$, such that $a_{i+1} \leq_A a_i$ for all $i \in \mathbb{N}$. 
In other words, in a well-founded poset, every element has only finitely many elements that precede it.}. 
Clearly, if $A$ is finite, then every partial order on $A$ is well-founded. 
In particular, the trivial poset $(\emptyset,\emptyset)$ is well-founded. 
It is not difficult to see that if a partial order has a well-founded linear extension, then it is well-founded. 
Conversely, 
if $\po=(A,\leq_A)$ is a well-founded poset, then $\po$ has a well-founded linear extension, see, \eg\cite{Downey-etal}. 
In particular, if $A$ is a countably infinite set\footnote{As it turns out, the partial orders that we consider in the body of the paper are all on finite or countably infinite sets.}
and $\po=(A,\leq_A)$ is well-founded, then one can prove\footnote{See, \eg https://math.stackexchange.com/questions/4194162/which-partial-orders-can-be-extended-to-a-copy-of-omega for a proof sketch,  as outlined here: 
\\
Given an enumeration $a_0, a_1, a_2, a_3, \ldots $ of $A$, recursively define a sequence $b_0, b_1, b_2\ldots$ in the following way.
Assuming that $b_0, \ldots , b_{k-1}$ have been defined for some $k \geq 0$, let $i$ be the least integer such that $a_i \not\in \{b_0, \ldots b_{k-1}\}$ and $\{a \mid a \leq_A a_i\} \subseteq \{b_0, \ldots b_{k-1}\}$.
Then $b_k = a _i$. 
Now define $\preceq_A$ by $a \preceq a'$ with $a \neq a'$ if $a = b_i$, $a' = b_j$ for some $i < j$. 
This $\preceq_A$ can be proved to be a linear extension of $\leq_A$.} 
that $\po$ has a well-founded linear extension on basis of an enumeration of $A$.
Be aware though, that even a well-founded partial order on a countably infinite set, may have a linear extension that is not well-founded. 

\begin{example}  
Consider $\po = (A,\Id_A)$ with $A = \{e_0, e_1, e_2, \ldots \}$ a countably infinite set. 
Then $\po$ is well-founded and $(A,\leq_A)$ with $\leq_A = \Id_A \cup \{(e_i,e_j) \mid i, j \in {\mathbb N} \text{ and } i < j \}$ is a well-founded linear extension of $\Id_A$. 
However, $(A,\preceq_A)$ with $\preceq_A \; = \; \Id_A \cup (\{(e_{2i},e_{2i-2})\mid i\geq 1\} \cup \{(e_0,e_1)\} \cup  \{(e_{2i+1},e_{2i+3}\mid i\geq 0\})^+$ is a linear extension of $\Id_A$ that is not well-founded.
\qed  
\end{example}

Let $\lo = (A,\leq_A)$ be a well-founded, linearly ordered set with $A \neq \emptyset$. 
Let $\lessdot_A$ be the direct predecessor relation of $\lo$:  
$a \lessdot_A b$ holds if $a \neq b$, $a \leq_A b$, and there is no $c \in A$ such that $a \neq c \neq b$ and $a \leq_A c \leq_A b$. 
Then $\lo$ defines the unique sequence $\sigma_{\lo} = a_1 a_2 \cdots $ such that $\{a_i \mid i \in \dom(\sigma_\lo) \} \subseteq A$ and  
$A = \{a_i \mid i \in \dom(\sigma_\lo) \}$ in case $\sigma_{\lo}$ is finite, 
and, for all $i,i+1 \in \dom(\sigma_\lo)$, $a_i\lessdot_A a_{i+1}$. 
Every $a \in A$ occurs at most once in $\sigma_{\lo}$, because $\leq_A$ as a partial order is acyclic. 
Since $\lo$ is well-founded and a linear order, it has a least element $a$ with $a \leq_A b$ for all $b \in A$. 
Clearly, this least element is $a_1$, the first element of $\sigma_\lo$. 

The \emph{language}, \ie the set of \emph{sequences defined by} a well-founded, partial order $\po$ is $\wdspo(\po) = \{\sigma_{\lo} \mid \lo \in \lex(\po) \mbox{ and }\lo \mbox{ is well-founded}\}$. 
We note that every well-founded poset defines at least one sequence. 

\begin{example}
Let $\po = (A,\leq_A)$ be a poset with $A = \{e_0, e_1, e_2, \ldots \} \cup \{d\}$ and $\leq_A = \{(d,d)\} \cup \{(e_i,e_j)\mid i, j \in {\mathbb N} \text{ and } i \leq j\}$. 
Then $\po$ is well-founded with minimal elements $e_0$ and $d$. 
\\
We set $\preceq^0_A \; = \; \leq_A \cup \; \{(d,e_i) \mid i \geq 0\}$, which is a well-founded linear extension of $\leq_A$ with minimal element $d$. 
Next, define, for all $k\in \mathbb{N}$, $k \geq 1$, the well-founded linear order $\preceq^k_A \; = \; \leq_A \cup \; \{(e_i,d) \mid 0\leq i < k \} \; \cup \; \{(d,e_i) \mid i \geq k\}$.
Each of these is a well-founded linear extension of $\leq_A$ with minimal element $e_0$. 
Finally, $\preceq_A \; = \; \leq_A \cup \; \{(e_i,d) \mid i \in {\mathbb N}\}$, which again is a well-founded linear extension of $\leq_A$. 
With $\lo = (A,\preceq_A)$, we have $\sigma_{\preceq_A} = e_0e_1 e_2 \cdots$, an infinite sequence with no occurrence of $d$.
Since $\po$ has no other (well-founded) linear extensions than the ones discussed above, $\wdspo(\po) = 
\{de_0e_1 e_2 \cdots \} \cup 
\{e_0e_1 e_2 \cdots e_k d e_{k+1} e_{k+2} \cdots \mid k \geq 0\} \cup 
\{e_0e_1 e_2 \cdots \}$ is the set of sequences defined by $\po$. 
\qed
\end{example}

Let $\po = (A,\leq_A)$ be a poset. 
A subset $B \subseteq A$ is said to be \emph{downward-closed} in $\po$ if for all $a \in A$ whenever $a \leq_A b$ for some $b \in B$, also $a \in B$. 
The \emph{downward closure} of $b \in A$ in $\po$ is the downward-closed set $\dc_\po(b)=\{a\in A \mid a \leq_A b \}$ consisting of all elements in $A$ that precede $b$ (including $b$ itself). 
It is not difficult to see that the prefixes of sequences defined by the linear extensions of a well-founded, partial order are consistent with the downward closure of their elements. 
Formally: 

\begin{lemma}\label{rem-pastlo}
    Let $\po = (A,\leq_A)$ be a partial order and $\lo$ a well-founded, linear extension of $\po$. 
    If $a_1 \cdots a_i \in \pref(\sigma_\lo)$ where $a_1,\ldots ,a_i \in A$ for some $i \geq 1$, then $\dc_\po(a_i) \subseteq \{a_1,\ldots ,a_i\}$.
\qed
\end{lemma}

Clearly, in case of an infinite partial order, each of its associated sequences is infinite and so each of these sequences is the limit of its finite prefixes.  
However, as the following result shows, also the converse holds: every word that can be obtained as the limit of finite prefixes of sequences of an infinite partial order $\po$ actually belongs to the sequences of $\po$. 
Note the contrast with Example~\ref{Ex-lim=more}.
As can be seen from its proof (in the Appendix), this result holds for all languages defined by partial orders, because for any word obtained as the limit of prefixes of words of a partial order $\po$ a linear extension $\lo$ of $\po$ can be constructed such that $w = \sigma_\lo$. 

\begin{lemma}
\label{Lem-seqlimpref}
    Let $\po = (A,\leq_A)$ be a well-founded, partial order with $A$ an infinite set. 
    Then $\wdspo(\po) = \llim(\Pref(\wdspo(\po)))$.   
\end{lemma}
\begin{proof}
    See Appendix~\ref{AppProofs}. 
\qed
\end{proof}

%%%%%%%%%%%%%%%%%%%%%%%%%%%%%%%%%%%%%%%%%%%%%%%%%%%%%
%%%%%%%%%%%%%%%%%%%%%%%%%%%%%%%%%%%%%%%%%%%%%%%%%%%%%

\section{Message Logs}
\label{Sect-MessLog}

%%%%%%%%%%%%%%%%%%%%%%%%%%%%%%%%%%%%%%%%%%%%%%%%%%%%%
%
A cross-organisational process is understood here as a global, distributed system consisting of a finite number of local business processes, executed in different organisations with the local processes communicating through asynchronous exchanges of typed messages over unidirectional, bilateral channels. 

\begin{example}
\label{Ex-OTC01}  
The so-called ``over the counter'' (OTC)
trading of securities involves high volumes of standardised message exchanges between financial institutions. 
Figure~\ref{PCollab01} shows part of the message exchanges between an Investment Firm and a Custodian (an intermediary party) involved in a securities transaction. 

\begin{figure}[!ht]
\begin{center}
 \begin{tikzpicture}
    \node[label={[align=center]above:\sffamily Investment\\\sffamily Firm}] (IF) at (0,0) {\includegraphics[scale=0.4]{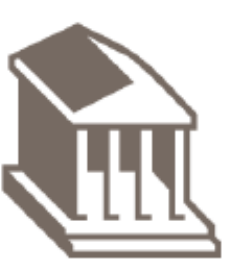}};
    \node[label=above:\textsf{Custodian}] (Cu) at (4,0) {\includegraphics[scale=0.4]{incld/tikz/bank}};

% dit is heel experimenteel. (lelijke woorden verwijderd)
\draw[->] (IF.20) --  pic[above,pos=0.40,yshift=2,scale=0.1] {envelope}  node[above,pos=0.60] {SI}  (Cu.160) ;
\draw[<-] (IF.-20) --  pic[above,pos=0.60,yshift=2,scale=0.1] {envelope}  node[above,pos=0.35] {SC}  (Cu.200) ;
    
\end{tikzpicture}
 % % Requires \usepackage{graphicx}
  % \includegraphics[scale=0.37]{incld/Pictures/Collab01}
    \caption{OTC settlement}\label{PCollab01}
\end{center}
\end{figure}

\noindent
The Investment Firm sends a \emph{Settlement Instruction}, instructing the Custodian to transfer securities, previously traded (sold or bought) on the OTC market, from or into the account of the Investment Firm. On completion of the instruction, the Custodian returns a \emph{Settlement Confirmation} to the Investment Firm. 
\qed
\end{example}

Every instantiation of a global process is represented by a \emph{case}. 
In Example~\ref{Ex-OTC01} above, the settlement of a transaction on the OTC market would be such instantiation and as a case typically be represented by a transaction identifier. 
A \emph{message} is a uniquely identifiable unit of information. 
Messages are exchanged between local processes that run together for a given case until a certain objective has been achieved (\eg the settlement of a specific transaction). 
If a message is exchanged between two processes, it is output from the sending process and input to the receiving process. 
Depending on its purpose, each message has a message format, or \emph{message type}. 
Throughout the paper, the special symbols $\out$ and $\inp$ indicate the sending and receiving, respectively, of messages.  
It is convenient to combine these concepts in the following background sets, fixed for the rest of this paper. 

\newenvironment{fminipage}% 
{\begin{Sbox}\begin{minipage}}{\end{minipage}
\end{Sbox}\fbox{\TheSbox}}\begin{center}
 
\begin{fminipage}{3.0in}  
$\PR$ is the universe of \emph{business processes}; 
\\
$\C$ is the universe of \emph{cases}; 
\\
$\I$ is the universe of \emph{messages};
\\
$\I \times \{\out,\inp\}$ is the universe of \emph{message events}; 
\\
$\mathcal{M}$ is the universe of \emph{message types}. 
\end{fminipage}
\end{center}

Let $e = (\m,\x)$ be a message event with $\m \in \I$ and $\x \in \{\inp,\out\}$. 
Then $\msg(e)=\m$ is the message of $e$. 
If $e = (\m,\out)$, we refer to $e$ as an \emph{output event} (representing the sending of the message $\m$) and if $e = (\m,\inp)$, it is an \emph{input event} (representing the receipt of $\m$).
We use the following (complement) notation: $\overline{\inp} = \out$ and $\overline{\out} = \inp$.
Thus, $\overline{\overline{\inp}}=\inp$ and $\overline{\overline{\out}}=\out$. 
Moreover, for $\m \in \I$ and $\x\in \{\inp,\out\}$, we set $\overline{(\m,\x)} = (\m,\overline{\x})$. 

Messages involved in a cross-organisational process belong to a case and have a message type. 

\begin{definition}\label{N07msgB0}%Let $\K \subseteq \I \times \{\inp,\out\}$. 
Let $\mathcal U \subseteq \I$ be a set of messages.  
\begin{itemize}[noitemsep,topsep=0pt]
  \item[(1)] 
  A \emph{case} function for 
   $\mathcal U$ is a function $f : \mathcal U \rightarrow \C$.  
   \item[(2)] 
   A \emph{type} function for $\mathcal U$ is a function $f : \mathcal U \rightarrow \mathcal M$.  
\qed
\end{itemize}
\end{definition}

Both case functions and type functions, and, in general, all functions defined for messages can be lifted to message events in the following way. 

\begin{definition}\label{Dinhoat} Let $\K \subseteq \I \times \{\inp,\out\}$ be a set of message events and let $f : \msg(\K) \rightarrow D$ for some set $D$. 
Then the \emph{enhancement} $\langle f \rangle : \msg(\K) \times \{\inp,\out\} \rightarrow D$ of $f$ is defined by $\langle f \rangle(\m,\x)=f(\m)$, for all $\m \in \msg(\K)$ and $\x \in \{\inp,\out\}$.
\qed
\end{definition}

Thus, message events inherit type and case from their messages.
As a consequence, given a message $\m$ and case function $\Cs$ and type function $\tp$ both defined for $\m$, we have that 
if $e = (\m,\x)$ with $\m \in \I$ and $\x \in \{\inp,\out\}$, then $\overline{e} = (\m,\overline{\x})$ and so $\langle \tp \rangle (e)= \tp(\m) = \langle \tp \rangle (\overline{e})$
and $\langle \Cs \rangle (e)= \Cs(\m) = \langle \Cs \rangle (\overline{e})$. 

Other relevant attributes of a message event are the local process that executes that event and the local time stamp associated to its execution. 
Note that, since communications take place between different processes, message events $e$ and $\overline{e}$ always belong to different processes. 
Moreover, no two different events executed by the same process (even if they belong to different cases) will have the same time stamp.

\begin{definition}\label{N07msgB0-2}
Let $\K \subseteq \I \times \{\inp,\out\}$. 
\begin{itemize}[noitemsep,topsep=0pt]
    \item[(1)] 
    A \emph{process function} for $\K$, is a function $p : \K \rightarrow \PR$, such that $p(e) \neq p(\overline{e})$, for all $e \in \K$, and $p(\K)$ is a finite set;
    \item[(2)] 
    Let $p$ be a process function for $\K$. 
    A \emph{local time function} for $p$ is a function $t : \K \rightarrow \mathbb{N}$, such that, for all distinct $e,e' \in \K$, $t(e) \neq t(e')$ whenever $p(e)=p(e')$. 
\qed
\end{itemize}
\end{definition}

We are now ready to define message logs as consisting of observations of communications, in the form of typed and time-stamped message events with an associated case, asynchronously exchanged over bilateral, unidirectional channels between a collection of local business processes. 

\begin{definition}\label{D07msgB0} 
A \emph{message log} is a pair $\Ka=(\K,\Att)$ such that 
\begin{itemize}[noitemsep,topsep=0pt]
    \item[(1)] 
    $\K \subseteq \I \times \{\inp,\out\}$ is a set of message events such that $e \in \K$ if and only if $\overline{e} \in\K$, and 
    \item[(2)] 
    $\Att =(\langle \Cs \rangle, \proc,$ $\tm, \langle \tp \rangle )$  is tuple providing the \emph{attributes of} $\Ka$ with 
    $\Cs$ a case function for $\msg(\K)$, 
    $\proc$ a process function for $\K$, 
    $\tm$ a local time function for $\proc$, and 
    $\tp$ a type function for $\msg(\K)$.
\qed
\end{itemize}
\end{definition}

\begin{example}
\label{Ex-OTC02}  (Ex.~\ref{Ex-OTC01} \ctd)  
Figure~\ref{PCollab02} illustrates the exchange of messages to settle a transaction. 
There are two local processes, namely $\Inv$ and $\Cus$. 
The downward pointing arrows represent the ordering of events along their local timelines. 

\begin{figure}[!ht]
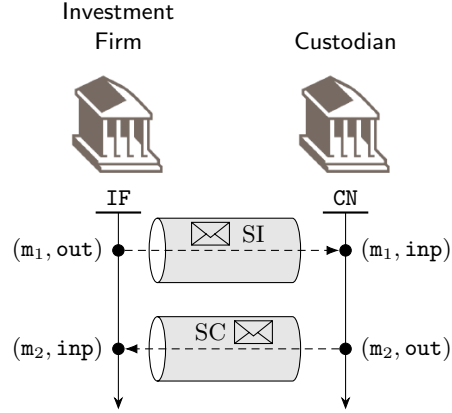

\begin{center}
  % % Requires \usepackage{graphicx}
  % \includegraphics[scale=0.35]{incld/Pictures/Collab02}
\begin{tikzpicture}
%% eerste proces x = 0
    \node[label={[align=center]above:\sffamily Investment\\\sffamily Firm}] (IF) at (0,3.5) {\includegraphics[scale=0.32]{incld/tikz/bank}};
\node at (0,2.5) {\tt IF} ;
\draw[thick] (0-0.3,2.3) -- (0+0.3,2.3) ;
\draw[time] (0,2.3) -- (0,-0.3) ;

\node[event,label=left:{$(\m_1,\out)$}] (P11) at (0,1.8) {} ;
\node[event,label=left:{$(\m_2,\inp)$}] (P12) at (0,0.5) {} ;

%% tweede proces x = 0 en dan verplaatsen met scope
\begin{scope}[shift={(3,0)}]
    \node[label=above:\textsf{Custodian}] (Cu) at (0,3.5) {\includegraphics[scale=0.32]{incld/tikz/bank}};
\node at (0,2.5) {\tt CN} ;
\draw[thick] (0-0.3,2.3) -- (0+0.3,2.3) ;
\draw[time] (0,2.3) -- (0,-0.3) ;

\node[event,label=right:{$(\m_1,\inp)$}] (P21) at (0,1.8) {} ;
\node[event,label=right:{$(\m_2,\out)$}] (P22) at (0,0.5) {} ;
\end{scope}

\tikzset{
     koker/.style={cylinder,draw=black,rotate=180,minimum height=2cm,minimum width=0.85cm,
     cylinder uses custom fill, cylinder body fill=black!10}
     }

\node[koker] at (1.5,1.8)  {} ;
\node[koker] at (1.5,0.5)  {} ;

\draw[message] (P11) --  (P21) ;
\draw[message] (P22) --  (P12) ;
\path (P11) --  pic[above,pos=0.40,yshift=2,scale=0.1] {envelope}  node[above,pos=0.60] {SI} (P21)  
(P22) --  pic[above,pos=0.40,yshift=2,scale=0.1] {envelope}  node[above,pos=0.60] {SC}   (P12) ;

\end{tikzpicture}
    \caption{OTC Settlement of transaction $T1$}
    \label{PCollab02}
\end{center}
\end{figure}

\noindent
Processes $\Inv$ and $\Cus$ exchange messages $\m_1$ (a Settlement Instruction) and $\m_2$ (a Settlement Confirmation). 
These exchanges are recorded in the message log $\Ka=(\K,\Att)$ with 
$\K=\{(\m_1,\out),(\m_1,\inp),(\m_2,\out),(\m_2,\inp) \}$ and $\Att =(\langle \Cs \rangle,$ $ \proc,$ $\tm, \langle \tp \rangle )$. 
\\
Let $\PR = \{\Inv,\Cus\}$. 
The process function $\proc$ for $\K$ is defined by $\proc(\m_1,\out)=\proc(\m_2,\inp)=\Inv$ and $\proc(\m_1,\inp)$ $=$ $\proc(\m_2,\out)$ $=$ $\Cus$.
\\
Let $T1$ be the identifier of this transaction. 
Then, $\Cs(\m_1)=\Cs(\m_2)=T1$ and so we have $\langle \Cs \rangle(e)=T1$  for all message events $e \in \K$. 
\\
The local time function $\tm$ assigns time stamps such that $\tm(\m_1,\out) < \tm(\m_2,\inp)$ and $\tm(\m_1,\inp)< \tm(\m_2,\out)$.
\\
Finally, with $\tp(\m_1)=SI$ and $\tp(\m_2)=SC$, we have 
$\langle \tp \rangle(\m_1,\x)=SI$ and $\langle \tp \rangle(\m_2,\x)=SC$ where $\x \in \{\inp,\out\}$. 
\qed
\end{example}

If $\Ka=(\K,(f_1, f_2 , f_3, f_4))$ is a message log, then $\card(\{f_2(e) \mid e \in \K \})$, the number of processes associated to $\Ka$, is referred to as the \emph{dimension} of $\Ka$.  
Note that by Definition~\ref{N07msgB0-2}(1), $\{f_2(e) \mid e \in \K \}$ is always a finite set. 
Moreover, by Definition~\ref{N07msgB0-2}(1) and 
Definition~\ref{D07msgB0}(1), if $\K \neq \emptyset$, then the dimension of $\Ka$ is always at least $2$. 

In what follows, we identify the processes associated to the message events from a message log by an integer from $[n]$, where $n$ is the dimension of the message log.

\begin{center}
\begin{fminipage}{4.0in}
\noindent
\emph{For the rest of this paper, 
$\Ka=(\K,\Att)$ is a fixed message log with 
\\
$\K$ a non-empty set of message events and 
\\
$\Att=(\langle \Cs \rangle, \proc,\tm, \langle\tp\rangle )$ 
where $\Cs$ and $\tp$ 
\\
are a case function and a type function, respectively, for $\msg(\K)$. 
%\\ $\proc$ is a process fumction for $\K$ and $\tm$ is a local time function for $\proc$. 
\\
The dimension of $\Ka$ is $n \geq 2$ and $\{ \proc(e) \mid e \in \K \}=[n]$. 
}
\end{fminipage}
\end{center}

%%%%%%%%%%%%%%%%%%%%%%%%%%%%%%%%%%%%%%%%%%%%%%%%%%%%%
%%%%%%%%%%%%%%%%%%%%%%%%%%%%%%%%%%%%%%%%%%%%%%%%%%%%%

\section{Structure of a message log}
\label{Sect-MLandPO}

%%%%%%%%%%%%%%%%%%%%%%%%%%%%%%%%%%%%%%%%%%%%%%%%%%%%%
%
In this section, we propose a relation to (partially) order the message events of message log $\Ka$. 
This relation is then used to specify a finitary and prefix-closed language for $\Ka$, that represents all sequences of message events defined by this partial order.  

Clearly, for each process of $\Ka$, the time stamps defined by its local time function determine a total order on all message events of this process. 
Recall that we assume that whenever a message event $e$ has an earlier (smaller) time stamp than another message event $e'$ belonging to the same process, then $e'$ does not causally precede $e$. 
In other words, the temporal order of two events from the same process as captured by its local time function, is either a necessary (causal) order determined by the (unknown) design of the process or a coincidental observation of their execution order. 
Note that for two events from the same process that belong to different cases, their allotted time stamps do not reflect a necessary ordering, as different cases are causally unrelated\footnote{This could have been formalised by having separate local time functions per case.}. 
Thus, we only relate message events that belong to the same case and we do this by combining the temporal (process) orders with the send-receive relations between output events and their corresponding input events. 

\begin{definition}\label{D07crl} 
  The \emph{event relation} of $\Ka$ is the relation
  $R_{\Ka} \; \subseteq \K \times \K$ such that $e \; R_{\Ka} \; e'$ 
  if and only if  
\\
(1) $\Cs(e)=\Cs(e')$, and 
\\
(2) 
either 
(a) $\proc(e)=\proc(e')$ and $\tm(e) < \tm(e')$ 
\\
or
(b) $e=(\m,\out)$ and $e'=(\m,\inp)$ for some $\m \in \msg(\K)$.
\qed
\end{definition}

Intuitively, the event relation of $\Ka$ captures for a given case an order of events that should be preserved in any global observation of its message events. 
However, as the following example shows, this relation is not necessarily acyclic. 

\begin{example}
\label{Fig-Cycle}
Let $\Ka = (\K,\Att)$ be a message log with set of message events 
$\K =$ $ \{(\m_1,\out), (\m_1,\inp), (\m_2,\out), (\m_2,\inp) \}$ and its attributes defined by $\Att =$ $ (\langle \Cs \rangle, \proc, \tm, \langle \tp \rangle )$.  
\\
The process function $\proc$ is defined by  $\proc(\m_1,\inp)=\proc(\m_2,\out)=P1$ and $\proc(\m_1,\out)=\proc(\m_2,\inp)=P2$. 
Note that $P1\neq P2$. 
\\
Let $\tm$ be such that 
$\tm(\m_1,\inp)=1$, $\tm(\m_2,\out)=2$, 
$\tm(\m_2,\inp)=1$, and $\tm(\m_1,\out)=2$. 
\\
If $\Cs(\m_1)=\Cs(\m_2)$, then the event relation $R_{\Ka}$ of $\Ka$ defines a cycle, see Figure~\ref{Fig-Cycle}(a).  

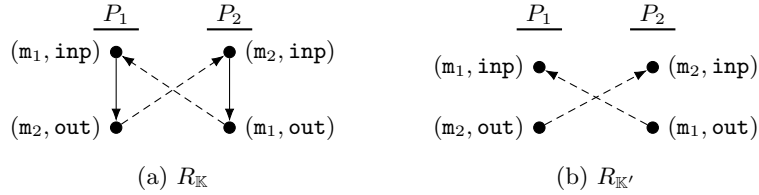
\begin{figure}[!ht]
\begin{subfigure}{0.45\textwidth}
\centering
\begin{tikzpicture}%[event/.style={draw,thick,circle,inner sep=0pt,minimum size=1.8mm}]

\node at (1,2.5) {$P_1$} ;
\draw[thick] (1-0.3,2.3) -- (1+0.3,2.3) ;
\node at (2.5,2.5) {$P_2$} ;
\draw[thick] (2.5-0.3,2.3) -- (2.5+0.3,2.3) ;

\node[event,label=left:{$(\m_1,\inp)$}] (P11) at (1,2) {} ;
\node[event,label=left:{$(\m_2,\out)$}] (P12) at (1,1) {} ;
\node[event,label=right:{$(\m_2,\inp)$}] (P22) at (2.5,2) {} ;
\node[event,label=right:{$(\m_1,\out)$}] (P21) at (2.5,1) {} ;

\draw[-Latex] (P11) edge (P12) (P22) edge (P21)  ;
\draw[message] (P12) edge (P22) (P21) edge (P11) ;
    
\end{tikzpicture}
\caption{$R_\Ka$}\label{fig:sub:XXX}
\end{subfigure}
\begin{subfigure}{0.45\textwidth}
\centering
\begin{tikzpicture}

\node at (1,2.5) {$P_1$} ;
\draw[thick] (1-0.3,2.3) -- (1+0.3,2.3) ;
%\draw[time] (1,2.3) -- (1,0+0.4) ;
\node at (2.5,2.5) {$P_2$} ;
\draw[thick] (2.5-0.3,2.3) -- (2.5+0.3,2.3) ;
%\draw[time] (2.5,2.3) -- (2.5,0+0.4) ;

\node[event,label=left:{$(\m_1,\inp)$}] (P11) at (1,1.8) {} ;
\node[event,label=left:{$(\m_2,\out)$}] (P12) at (1,1.0) {} ;
\node[event,label=right:{$(\m_2,\inp)$}] (P22) at (2.5,1.8) {} ;
\node[event,label=right:{$(\m_1,\out)$}] (P21) at (2.5,1.0) {} ;

\draw[message] (P21) edge (P11)  (P12) edge (P22) ;

\end{tikzpicture}
\caption{$R_{\Ka'}$}\label{fig:sub:YYY}
\end{subfigure}
  
    \caption{$(\m_1,\inp),(\m_2,\out),(\m_2,\inp),(\m_1,\out)$ form a cycle in $R_\Ka$ but not in $R_{\Ka'}$.} 
\label{Pycle}
\end{figure}

\noindent
However, if $\Ka' = (\K,\Att')$ is a message log obtained from $\Ka$ by replacing $\Cs$ by a case function $\Cs'$ such that $\Cs'(\m_1) \neq \Cs'(\m_2)$, the event relation $R_{\Ka'}$ is acyclic: $((\m_1,\inp),(\m_2,\out))$ and $((\m_2,\inp),(\m_1,\out))$ are not included in $R_{\Ka'}$, \cf Figure~\ref{Fig-Cycle}(b), even though $\tm(\m_1,\inp) < \tm(\m_2,\out)$ and $\tm(\m_2,\inp) < \tm(\m_1,\out)$. 
\qed
\end{example}

It is a straightforward observation that every non-empty cycle of $R_{\Ka}$ consisting of message events (which by definition all belong to the same case), contains an input event and an output event with the same underlying message. 
Consequently, cycles have an input event that precedes its corresponding output event, something which cannot happen in a message log that properly incorporates the execution order of its events. 
Hence 

\begin{definition}\label{Def-sound}
$\Ka$ is \emph{sound} if $R_{\Ka}$ is an acyclic relation. 
\qed
\end{definition}

If $R_{\Ka}$ is acyclic, then its transitive closure $R_{\Ka}^+ $ is antisymmetric. 
By construction, $R_{\Ka}^* $ is a transitive and reflexive relation. 
This leads to the following observation as a direct consequence of Definition~\ref{Def-sound}.

\begin{lemma}\label{Lsoundpo}
$R_{\Ka}^* $ is a partial order if and only if $\Ka$ is sound. 
\qed
\end{lemma}

\begin{center}
\begin{fminipage}{3.8in}
\emph{
For the rest of this paper, $\Ka$ is a sound message log. 
\\
All message logs considered from here are assumed to be sound.
}
\end{fminipage}
\end{center}

Since $\Ka$ is fixed, the subscript $\Ka$ can be dropped from $R_{\Ka}^*$. 
To further simplify our notation, we write $\preceq$ instead of $R^*$. 
So, from now on, we investigate the poset $\po_\K = (\K, \preceq)$ defined by the sound message log $\Ka=(\K,\Att)$. 
In addition, given a subset $C$ of $\K$, we write $\po_C = (C,\preceq_C)$ for its induced sub-poset in $(\K,\preceq)$.

Next, we investigate the sequences of message events defined by $\po_\K$ that belong to the same case.
Note that message events from different cases are not related in $\po_\K$. 
As an initial, general, observation, we establish that every subset of $\K$ consisting of message events belonging to the same case, induces a partial order of which \emph{all} linear extensions are well-founded.

\begin{lemma}\label{Kwfd} 
Let $C \subseteq \K$ be such that $\Cs(e) = \Cs(e')$, for all $e,e' \in C$. 
Let $\lo = (C, \preceq_\lo) \in \lex(\po_C)$.
Then $\lo$ is well-founded.
\end{lemma}
\begin{proof}
    Let $D \subseteq C$ be infinite 
    (if $C$ is finite, there is nothing to prove). 
    Let, for all $\ell \in [n]$, $D_\ell = \{e \in D \mid \proc(e) = \ell\}$ be the set of elements of $D$ that belong to process $\ell$. 
    Recall that $\tm: \Ka \rightarrow {\mathbb N}$ is such that for all distinct message events $e,e'$, $\tm(e) \neq \tm(e')$ whenever $\proc(e) = \proc(e')$. 
    Consequently, for all $\ell \in [n]$ and for every $e\in D_\ell$, there can be only finitely many $e' \in D_\ell$ such that $\tm(e') < \tm(e)$ and hence $e' \preceq_\lo e$ (\cf Definition~\ref{D07crl}). 

    Let $\ell \in [n]$. 
    Since $\preceq_\lo$ is a linear order on $C$, also its restriction to  $D_\ell$ is a linear order.
    Hence, if $D_\ell$ is not empty, it has a unique minimal element w.r.t. $\preceq_\lo$, which we denote by $d_\ell$. 

    Let us assume that none of the $d_\ell$ thus defined is minimal in $D$ w.r.t. $\preceq_\lo$. 
    Hence, for all $\ell \in [n]$ and for all $d_\ell$, there is an $e \in D$ with $e\not\in D_\ell$ such that $e \preceq_\lo d_\ell$. 
    It then follows from Definition~\ref{D07crl} that all $d_\ell$ are input events. 
   More precisely, for all $\ell\in [n]$, there exist an $i_\ell \in [n]$ and an output event $e_\ell \in D_{i_\ell}$ such that $i_\ell \neq \ell$ and $e_\ell = \overline{d_\ell}$. 
    This implies that, for all $\ell \in [n]$, the minimal element $d_{i_\ell}$ of $D_{i_\ell}$ precedes $d_\ell$ in $\lo$ since $d_{i_\ell} \preceq_\lo e_\ell \preceq_\lo d_\ell$. 
    Repeating the above argument leads to the conclusion, that  - 
    as we have only a finite number $n$ of processes - there exists a cycle $d_i \preceq_\lo d_j \preceq_\lo d_i$ for some $i,j \in [n]$.
    This is in contradiction with $\lo$ being a partial order. 
    Hence there must exist at least one $\ell \in [n]$ such that $d_\ell$ is minimal in $D$ w.r.t. $\lo$. 

    We conclude that every nonempty subset of $C$ has a minimal element w.r.t. $\preceq_\lo$ and so $\lo$ is well-founded. 
\qed
\end{proof}

Since a partial order is well-founded if it has a linear extension that is well-founded (\cf Section~\ref{Sect-Prelim}), we have as an immediate corollary that $\po_C$ is well-founded for every $C \subseteq \K$ consisting of message events in $C$ that all belong to the same case. 
Hence, by the definition of well-foundedness and because message events that belong to different cases are not ordered by $\preceq$, every message event of $\Ka$ has a finite past. 
In what follows, we consider message events together with their past, \ie all message events that precede them in $\po_\K$. 
Thus, we are interested in sets of message events belonging to the same case together with their downward closure in $\po_\K$. 

\begin{definition}\label{Def-Con}
\label{Dconvs0} 
A \emph{conversation} in $\Ka$ is a subset $C \subseteq \K$ such that 
\\
(1) $\Cs(e)=\Cs(e')$ for all $e,e' \in C$, and 
\\
(2) $C$ is downward-closed in $(\K,\preceq)$. 
\qed
\end{definition}

Note that by Definition~\ref{D07crl}(1), all message events in the downward closure $\dc_{\po_\K}(e)$ of a message event $e$ have the same case. 
Hence, requirements (1) and (2) in Definition~\ref{Def-Con} are independent. 

A conversation $C$ captures (an initial part of) behaviour of the system for a certain case in the form of the partial order $\po_C$. 
By Lemma~\ref{Kwfd}, all linear extensions of $\po_C$ are well-founded.
Since $\preceq_C \, = \bigcap\{\preceq_\lo \mid (C,\preceq_\lo) \in \lex(\po_C\}$, 
\cf Section~\ref{Sect-Prelim}, we have $\wdspo(\po_C)$, comprising the sequences corresponding to the linear extensions of $\po_C$, as a full representation of the sequential behaviour captured by $C$. 
Since in the case of a finite conversation, all its elements appear in each of its sequences, the sets of sequences defined by different finite conversations are always disjoint. 

\begin{lemma}
\label{Lem-uniqueC}
Let $C$ and $C'$ be finite conversations in $\Ka$ such that $C \neq C'$. 
Then $\wdspo(\po_{C}) \cap \wdspo(\po_{C'}) = \emptyset$. 
\qed
\end{lemma}

However, as $\K$ may be infinite, conversations may be infinite as well.
This implies that $\wdspo(\po_C)$ is not necessarily a finitary language. 
Therefore, we will now argue that we can restrict the representation of a system's behaviour to finite conversations. 
First, we demonstrate that the sequences defined by the finite subconversations of a conversation $C$ are exactly the (finite) prefixes of $\wdspo(\po_C)$.  

\begin{lemma}
\label{Lem-prefseq=seqpref}
    Let $C$ be a conversation in $\Ka$. Let $w \in \K^*$. 
    \\
    Then $w \in \Pref(\wdspo(\po_C))$ if and only if there exists a finite conversation $C' \subseteq C$ such that $w \in \wdspo(\po_{C'})$. 
\end{lemma}
\begin{proof}
(only-if-direction) 
    Assume that $w \in \Pref(\wdspo(\po_C))$. 
    Let $\lo = (C,\leq_\lo) \in \lex(\po_C)$ be such that $w \in \pref(\sigma_\lo)$.
    Thus $\ab(w) \subseteq C$.
    Let $C' = \ab(w)$. 
    We now set out to prove that $C'$ is a finite conversation in $\Ka$ and $w = \sigma_{\lo'}$ for a linear extension $\lo'$ of $\po_{C'} = (C',\preceq_C')$. 
    Clearly, $C'$ is a finite set, as $w$ is a finite prefix of $\sigma_\lo$.
    Since $C$ is a conversation, $\Cs(e) = \Cs(e')$ for all $e,e' \in C' \subseteq C$.
    What remains to be shown is that $C'$ is downward-closed in $(\K,\preceq)$.
    Since $C$ is downward-closed, $\dc_{\po_\K}(e) = \dc_{\po_C}(e)$ for all $e \in C$. 
    Moreover, $\lo$ is well-founded by Lemma~\ref{Kwfd}.
    Consequently, by Lemma~\ref{rem-pastlo}, $\dc_{\po_\K}(w(i)) = \dc_{\po_C}(w(i)) = \{e\in \K \mid e \preceq_C w_i \} \subseteq \{w(1), \ldots , w(i)\} \subseteq C'$ for all $i \in \dom(w)$. 
    Hence, $C'$ is downward-closed in $(\K,\preceq)$ thus a finite conversation in $\Ka$. 
\\
    Consider now $\po_{C'} = (C',\preceq_{C'})$ and define $\sqsubseteq_w = \{(w(i),w(j)) \mid i,j \in \dom(w), i \leq j \}$.
    Clearly, $\sqsubseteq_w$ is a linear order on $C'$. 
    Moreover, for all $c,d \in C'$, we have $c \sqsubseteq_w d$ if and only if $c \leq_\lo d$. 
    Hence  $\sqsubseteq_w = \leq_\lo \cap \; (C' \times C')$. 
    Thus, since  $\lo$ is a linear extension of $\preceq_C$, we have $\preceq_C' \; = \; \preceq_C \cap \; (C' \times C') \subseteq \; \leq_\lo \cap \; (C' \times C') = \; \sqsubseteq_w$. 
    Consequently, $\lo' = (C',\sqsubseteq_w) \in \lex(\po_{C'})$. 
    As it is easy to see that $w = \sigma_{\lo'}$, we are done.
    
\smallskip
(if direction)
    Let $C' \subseteq C$ be a finite conversation in $\Ka$ such that $w \in \wdspo(\po_{C'})$.  
    Let $\lo' = (C',\leq_{\lo'}) \in \lex(\po_{C'})$ be such that $w = \sigma_{lo'}$. 
    Consider $\sqsubseteq_C \; = \; \leq_{\lo'} \cup \preceq_C \cup \{(c,d) \mid c \in C', d \in C\setminus C'\}$, the relation over $C$ capturing the linear order underlying $w$, combined with the partial order $\preceq_C$ induced by $C$, and extended in such a way that all message events in $C'$ precede all other message events in $C$. 

    First, we establish that $\sqsubseteq_C$ is a partial order on $C$. 
\\   
    Reflexivity of $\sqsubseteq_C$ follows immediately from the reflexivity of $\preceq_C$. 
\\
    To establish  transitivity of $\sqsubseteq_C$, we first observe that there do not exist $d \in C\setminus C'$ and $e \in C'$ such that $d \sqsubseteq_C e$.
    Indeed, suppose to the contrary that $d \in C\setminus C'$ and $e \in C'$ are such that $d \sqsubseteq_C e$. 
    Then by the definition of $\sqsubseteq_C$ it must be the case that $d \preceq_C e$. 
    However, $C'$, as a conversation in $\K$, is downward-closed in $(\Ka,\preceq)$.
    Since $(C,\preceq_C)$ is the sub-poset in $(\Ka,\preceq)$ induced by $C$, $C'$ is also downward-closed in $(C,\preceq_C)$. 
    So, $d \in C'$, a contradiction.
    Now, let $a,b,c \in C$ be such that $a \sqsubseteq_C b$ and $b \sqsubseteq_C c$. 
    Then, from the above it follows that if $a \in C \setminus C'$ then also $b \in C \setminus C'$. 
    Similarly, if $b \in C \setminus C'$ then also $c \in C \setminus C'$. 
    Furthermore, if $a \in C'$ and $c \in C\setminus C'$, then $a\sqsubseteq_C c$ by definition. 
    Consequently, there are only two cases left to investigate, viz., $a,b,c \in C'$ and $a,b,c \in C\setminus C'$. 
    In the former case, we first observe that $\preceq_C \cap \; (C' \times C') \subseteq \;\leq_{\lo'}$ because $\lo'$ is a linear extension of $\po_{C'}$ and $\preceq_{C'} \; = \; \preceq_C \cap \; (C' \times C')$. 
    Thus $a \sqsubseteq_C b$ and $b \sqsubseteq_C c$ with $a,b,c \in C'$ together imply $a \leq_{\lo'} b$ and $b \leq_{\lo'} c$. 
    Then $a \leq_{\lo'} c$ by the transitivity of $\leq_{\lo'}$, and so $a \sqsubseteq_C c$. 
    In the latter case, it follows from $a \sqsubseteq_C b$ and $b \sqsubseteq_C c$ that    $a \preceq_C b$ and $b \preceq_C c$.  
    Hence by the transitivity of $\preceq_C$, we have $a \preceq_C c$ and so $a \sqsubseteq_C c$. 
\\
    To prove that $\sqsubseteq_C$ is antisymmetric, assume that we have $a,b \in C$ such that $a \sqsubseteq_C b$ and $b \sqsubseteq_C a$. 
    If both $a,b \in C'$, then $a \leq_{\lo'} b$ and $b \leq_{\lo'} a$. 
    Hence $a = b$ by the antisymmetry of $\leq_{\lo'}$.
    Similarly, if both $a,b \in C \setminus C'$, then $a \preceq_C b$ and $b \preceq_C a$, which implies that $a = b$ by the antisymmetry of $\preceq_C$. 
    If $a \in C'$ and $b \in C \setminus C'$, then it must be that $b \preceq_C a$. 
    Then, as above, since $C'$ is a conversation in $\Ka$, it is downward-closed in $(\K,\preceq)$ and thus also $b \in C'$, a contradiction.

   Next, we consider a linear extension $\lo = (C,\leq_\lo)$ of $(C,\sqsubseteq_C)$. 
    Let $v = \sigma_\lo$. 
    We prove that $w \in \pref(v)$. 
\\
    Suppose $w \not\in \pref(v)$. 
    Let then $i \in \dom(w)$ be such that $w(j) = v(j)$ for all $j < i$, and $w(i) \neq v(i)$.  
    Let us assume first that $v(i) \in C'$. 
    Then $v(i) = w(k) $ for some $k > i$, and $w(k) \leq_\lo w(i)$. 
    Since $w(i) \leq_{\lo'} w(k)$ also $w(i) \sqsubseteq_C w(k)$ which implies that $w(i) \leq_\lo w(k)$ and thus $w(i) = w(k)$, a contradiction.
    Thus it must be the case that $v(i) \in C \setminus C'$. 
    Then $w(i) \sqsubseteq_C v(i)$ by the definition of $\sqsubseteq_C$. 
    Furthermore, $v(i)\leq_\lo w(i)$ because $v = \sigma_\lo$.
    Hence $v(i) \sqsubseteq_C w(i)$.
   By the antisymmetry of $\subseteq_C$, $v(i) = w(i)$, again a contradiction.
    Therefore $w \in \pref(v)$ and thus $w \in \Pref(\wdspo((C,\sqsubseteq_C)))$.

    Since $\preceq_C \; \subseteq \; \sqsubseteq_C$ implies $\lex((C,\sqsubseteq_C)) \subseteq \lex(\po_C)$, it follows that $w \in \Pref(\wdspo(\po_C))$ which concludes the proof. 
\qed
\end{proof}

It is interesting to observe here that every downward closed subset (`prefix') of a conversation is a subconversation of $C$. 
Conversely, any finite subconversation $C'$ of a conversation $C$ is a prefix of $C$ because it is downward closed. 
In other words, the finite subconversations of conversation $C$ are its prefixes and by the above lemma, the prefixes of the sequences defined by $C$ are the sequences defined by the prefixes of $C$.

As an immediate consequence of Lemma~\ref{Lem-prefseq=seqpref} and Lemma~\ref{Lem-seqlimpref}, we have that for every infinite conversation, the infinitary language defined by its induced partial order, can be obtained as the limit of the union of the finitary languages defined by the partial orders induced by its finite subconversations. 

\begin{theorem}
\label{Thm-limConv}
Let $C$ be an infinite conversation in $\Ka$. 
Then 
$\wdspo(\po_C) =$
\\
$ 
\llim(\bigcup\{\wdspo(\po_{C'}) \mid C'  
\text{ a finite subconversation of } C \})$.
\qed
\end{theorem}

Note that for every infinite conversation $C$, its language $\wdspo(\po_C)$ consists of infinite words only. 
On the other hand, for a finite conversation $C'$,  $\wdspo(\po_C')$ is a finite set comprising only finite words. 
By Theorem~\ref{Thm-limConv}, all sequences defined by the linear extensions of an infinite conversation are approximated by the infinite set consisting of the finite sequences of the linear extensions of its finite subconversations.
Thus, we can now define a finitary language for the message log $\Ka$ as the union of all languages of all its finite conversations (thus covering all cases). 

\begin{definition}\label{Def-EL}
The event language of $\Ka$ is the set 
$\Ev(\Ka) = \bigcup\{\wdspo(\po_C) \mid C$ $\text{a finite conversation in } \Ka\}$ comprising all sequences of message events defined by the linear extensions of the finite conversations in $\Ka$. 
    \qed
\end{definition}

By Definition~\ref{Def-EL} and Lemma~\ref{Lem-prefseq=seqpref}, $\Ev(\Ka)$ is a finitary language consisting of all prefixes of all words from the languages defined by the partial orders of its conversations. This observation leads to 

\begin{theorem}\label{Thm-Seq=finpref}
    $\Ev(\Ka)$ is a finitary, prefix-closed language.  
\qed
\end{theorem}

%%%%%%%%%%%%%%%%%%%%%%%%%%%%%%%%%%%%%%%%%%%%%%%%%%%%%
%%%%%%%%%%%%%%%%%%%%%%%%%%%%%%%%%%%%%%%%%%%%%%%%%%%%%

\section{Identifying actions in a message log}
\label{Sect-DCA} 

%%%%%%%%%%%%%%%%%%%%%%%%%%%%%%%%%%%%%%%%%%%%%%%%%%%%%
%
To convert message log $\Ka$ into a global event log with an associated distributed communicating alphabet, message events have to be interpreted as  executions of actions of local processes. 
In this section, we investigate how to consistently identify actions in a message log, leading to a labelling of events with action names that can be used to define a distributed communicating alphabet (Definition~\ref{Def-DCA} in Appendix~\ref{AppSInets}). 

Recall that $\Ka = (\K,\Att)$ with $\Att =(\langle \Cs \rangle, \proc,\tm, \langle\tp\rangle )$. 
The number of processes (the dimension) of $\Ka$ is $n \geq 2$ and $\{ \proc(e) \mid e \in \K \}=[n]$. 
In what follows, we let 
$\K_{i}=\{e \in \K \mid \proc(e) = i\}$ where $i \in [n]$, be the set of message events in $\K$ that belong to the $i$-th process. 
Note that, for all $i \in [n]$, $\K_{i} \neq \emptyset$ because $\{ \proc(e) \mid e \in \K \}=[n]$.
Let 
$\K_{\out}=\K \; \cap \; (\I \times \{\out\})$ and $\K_{\inp}=\K \; \cap \; (\I \times \{\inp\})$ be the set of output events and the set of input events of $\K$, respectively.
Since by Definition~\ref{D07msgB0}, 
$e \in \K$ if and only if $\overline{e} \in\K$, we have $\K_{\out} \neq \emptyset$ and $\K_{\inp} \neq \emptyset$. 
Thus, both $\{\K_1, \ldots , \K_n\}$ and $\{\K_\out, \K_\inp\}$ are partitions of $\K$.
Finally, we rephrase the complement notation introduced in Section~\ref{Sect-MessLog}, as a function: 

\begin{definition}\label{Def-Mex}
The \emph{message exchange function} of $\K$ is the function $\cp: \K \rightarrow \K$ defined by $\cp(e)=\overline{e}$, for all $e \in \K$. 
\qed
\end{definition}

The function $\cp$ is clearly a complement function.
It also preserves types.
This follows from 
$\langle \tp \rangle (\cp(e)) = \langle \tp \rangle (\overline{e}) = \tp (\msg(\overline{e})) = \tp (\msg(e)) = \langle \tp \rangle (e)$.

\medskip
Our aim is to associate a distributed communicating alphabet with $\Ka$. 
So far, we have several entities defined by $\Ka$ which together closely resemble a $\DA$ (\cf Definition~\ref{Def-DCA}), namely the sets $\K_1, \ldots, \K_n$ and $\K_{\out},\K_{\inp}$, with both $[\K_1, \ldots, \K_n]$ and $[\emptyset,\K_{\out},\K_{\inp}]$ consisting of mutually disjoint sets such that $\bigcup_{i \in [n]} \K_i=\K=\K_{\out} \cup \K_{\inp}$ and for all $i \in [n]$, $\K_{i} \neq \emptyset$, as well as the functions $ \langle \tp \rangle$ and $\cp$. 
The function $\langle \tp \rangle$ assigns message types to message events and the complement function $\cp$ preserves message types. 
The tuple $([\K_1, \ldots, \K_n], [ \emptyset,\K_{inp},\K_{out}],\langle \tp \rangle, \cp)$ is however not an $n$-DCA in case $\K$ is an infinite set. 
More fundamentally, message logs consist of unique message events rather than multiple occurrences of input or output actions. 
Thus we propose labelling functions to identify message events that represent different occurrences of the same action of a local process. 
These labelling functions should preserve the partitions defined by processes and input and output events.
Furthermore, whenever events are identified by the labelling function, then so are their complements. 

\begin{definition}\label{Dafngen} 
An \emph{action function} (for $\K$) is a function 
$f : \K \rightarrow D$ with $D$ a set of (action) labels, which satisfies the following two properties:  
\\ 
\emph{Disjointness} (\emph{DP}): 
both $\{ f(\K_i) \mid i \in [n] \}$  and $\{f(\K_{\inp}),f(\K_{\out})\}$ are partitions of $f(\K)$;  
\\ 
\emph{Complement} (\emph{CP}): for all message events $e,e' \in \K$, if $f(e)=f(e')$, then $f(\cp(e))=f(\cp(e'))$. 
\qed
\end{definition} 

\begin{example}\label{Ex-ActFn} 
Let $\Ka=(\K,\Att)$ be a message log, $n \geq 2$ its dimension, and $f$ an action function for $\K$. 
This function maps certain output events of process $i$ to the action label $\mbox{\bf{s}}$ as depicted in Figure~\ref{PctionFn01}. 
We assume that $(\m_1,\out),$ $\ldots,$ $(\m_\ell,\out)$ are \emph{all} (here a finite number $\ell \geq 1$) output events of process $i$ that are mapped to action label $\mbox{\bf{s}}$. 
By the disjointness property (DP) of $f$, there are no message events belonging to another process that are also mapped to $\mbox{\bf{s}}$. 
Now consider $(\m_1,\inp),$ $\ldots,$ $(\m_\ell,\inp)$, the complements under $\cp$ of $(\m_1,\out),$ $\ldots,$ $(\m_\ell,\out)$, respectively. 
By the complement property (CP) of $f$, each of $(\m_1,\inp),$ $\ldots$ $,(\m_\ell,\inp)$ is mapped to the same action label, here $\mbox{\bf{r}}$. 
Moreover, again by DP, these input events all belong to the same process $j$ with $j \neq i$. 
Thus, as depicted in Figure~\ref{PctionFn01}, actions $\mbox{\bf{s}}$ and $\mbox{\bf{r}}$ represent a communication channel between process $i$ and process $j$ with $\mbox{\bf{s}}$ representing the sending of an $m_i$, $i \in [\ell]$ and $\mbox{\bf{r}}$ representing the receiving of an $m_i$, $i \in [\ell]$. 

\begin{figure}[!ht]
\begin{center}

\begin{tikzpicture}[event/.style={draw,fill,circle,inner sep=0pt,minimum size=0.8mm}] 
%% --
%% x-as links=1 rechts=6 midden=3.5
%% y-as midden-3.5

%% proces links
\node at (1,6.7) {Process $i$} ;

\draw (1,3.5) ellipse (0.8 and 3);
\draw (1,3.5) ellipse (0.6 and 2);

\foreach \x in {1,2,...,12} 
    \node[event] at (1,\x/2+0.25) {} ;
\node[event,label=below:{$(m_1,\out)$}] (1out) at (1,5+0.25) {} ;
\node[event,label=above:{$(m_\ell,\out)$}] (nout) at (1,2-0.25) {} ;

%% proces rechts
\node at (6,6.7) {Process $j$} ;

\draw (6,3.5) ellipse (0.8 and 3);
\draw (6,3.5) ellipse (0.6 and 2);
\foreach \x in {1,2,...,12} 
    \node[event] at (6,\x/2+0.25) {} ;
\node[event,label=below:{$(m_1,\inp)$}] (1inp) at (6,5+0.25) {} ;
\node[event,label=above:{$(m_\ell,\inp)$}] (ninp) at (6,2-0.25) {} ;

%% out-in pijlen
\draw[mexmap] (1out) to[bend left=10] node[above] {\cp} (1inp) ;
\draw[mexmap] (nout) to[bend left=-10] (ninp) ;

%% cylinder
\coordinate[label={[label distance=2]left:\textbf{s}}] (sa) at (3.5-0.7,3.5) ;
\coordinate[label={[label distance=2]right:\textbf{r}}] (ra) at (3.5+0.7,3.5) ;
\draw[very thick,dashed,-Latex] (sa) to (ra) ;
%% https://tex.stackexchange.com/questions/86535/generate-simple-cylindrical-shape-with-text-in-latex-tikz
%% arc verschuiven want beginpunt op arc, niet centrum ellips
\draw (3.5-0.7,3.5) ellipse (0.15 and 0.4);
\draw (3.5+0.7,3.5-0.4) arc (-90:90:0.15 and 0.4);
\draw[dotted] (3.5+0.7,3.5+0.4) arc (90:270:0.15 and 0.4);
\draw (3.5-0.7,3.5+0.4) -- (3.5+0.7,3.5+0.4) ;
\draw (3.5-0.7,3.5-0.4) -- (3.5+0.7,3.5-0.4) ;
\fill [opacity=0.2] (3.5-0.7,3.5+0.4) -- (3.5+0.7,3.5+0.4)  arc (90:-90:0.15 and 0.4) -- (3.5-0.7,3.5-0.4)  arc (-90:90:0.15 and 0.4) ;
\tikzset{farrow/.style={-Latex}}
%% met de hand aangrijpingspunt ellips aangepast
% top = centrum ellips (1,3.5) + (0,2)  
\draw[farrow] (1+0.2, 3.5+2-0.1) to node[above,pos=0.7] {$f$} (3.5-0.7,3.5+0.4);
\draw[farrow] (1+0.2, 3.5-2+0.1) to node[below,pos=0.7] {$f$} (3.5-0.7,3.5-0.4);
% rechter ellips
\draw[farrow] (6-0.2, 3.5+2-0.1) to node[above,pos=0.7] {$f$} (3.5+0.7,3.5+0.4);
\draw[farrow] (6-0.2, 3.5-2+0.1) to node[below,pos=0.7] {$f$} (3.5+0.7,3.5-0.4);
\end{tikzpicture}

\caption{Action function $f$.}
\label{PctionFn01}
\end{center}
\end{figure}
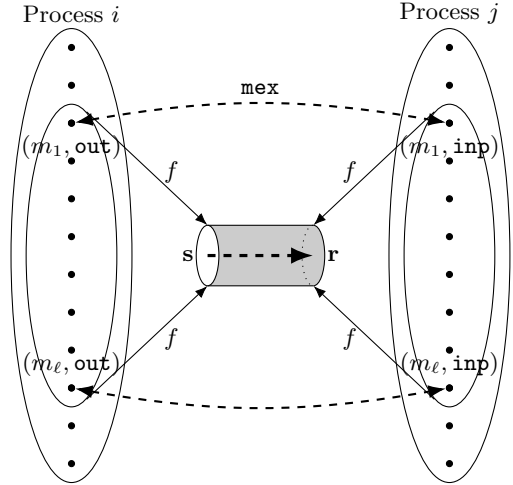

\noindent
Note that one could also interpret the situation depicted in Figure~\ref{PctionFn01} as resulting from a labelling of the messages $\m_1,$ $\ldots,$ $\m_\ell$ by say $\mbox{\bf{a}}$, which labelling is then lifted to the output events $(\m_1,\out),$ $\ldots,$ $(\m_\ell,\out)$ as $\mbox{\bf{s}}$ (``send $\mbox{\bf{a}}$'' or $(\mbox{\bf{a}},\out)$) and to the input events $(\m_1,\inp),$ $\ldots,$ $(\m_\ell,\inp)$ as $\mbox{\bf{r}}$ (``receive $\mbox{\bf{a}}$'' or $(\mbox{\bf{a}}, \inp)$). 
\qed
\end{example} 

As suggested in Example~\ref{Ex-ActFn}, another approach could be to start from a labelling of the messages and then lift this labelling to message events by a reference to their role as output or input event.
This gives rise to the following definition. 

\begin{definition}\label{Dinhoat-2}
Let $g : \msg(\K) \rightarrow D$ where $D$ is a set of (message) labels. Then the \emph{event extension} $g_{\ee}: \K \rightarrow D \times \{\inp,\out\}$ of $g$ to message events, is defined by $g_{\ee} ((\m,\out))= (g(\m), \out)$ and $g_{\ee}((\m,\inp)) = (g(\m),\inp)$ for all $\m \in \msg(\K)$.
\qed
\end{definition}

As our next result shows, a message labelling function defines through its event extension an action function, provided the message labels agree with the communication channels between the processes: two messages that get the same label should have the same sending and receiving processes. 
This is formalised in the next definition. 

\begin{definition}\label{Dchprop-bis} 
(1) Let $\m \in \msg(\K)$.
Then the \emph{send/receive pair} of $\m$ is 
$\Ch(\m) = 
(\proc((\m,\out)), \proc((\m,\inp)))$. 
\\
(2)
Let $g : \msg(\K) \rightarrow D$ for an alphabet $D$. 
Then $g$ has the \emph{channel property} if, for all $\m,\m' \in \msg(\K)$, $g(\m) = g(\m')$ implies $\Ch(\m)=\Ch(\m')$. 
\qed
\end{definition}

\begin{theorem}\label{Thm-Chp}
Let $g : \msg(\K) \rightarrow D$ for some set $D$. 
Then $g_{\ee}$ is an action function if and only if $g$ has the channel property. 
\end{theorem}

\begin{proof} 
(only if)
    Let  $\m, \m' \in \msg(\K)$ be such that $g(\m)=g(\m')$. 
    It follows that $g_{\ee}((\m, \out)) = (g(\m),\out) = (g(\m'),\out) = g_{\ee}((\m', \out))$ and, similarly, that $
    g_{\ee}(\m, \inp) = g_{\ee}(\m', \inp)$. 
    Since $g_{\ee}$ is an action function, it satisfies DP, the disjointness property. 
    Consequently, we have $\proc((\m,\out)) = \proc((\m',\out))$ and $\proc((\m,\inp)) = \proc((\m',\inp))$. 
    Hence, $\Ch(\m)=\Ch(\m')$.  
    We conclude that $g$ has the channel property. 

    (if)
    By assumption, $g$ has the channel property. 
    Firstly, we establish the disjointness property DP of $g_{\ee}$. 
    By definition, $g_{\ee}(\K_\inp) = g(\msg(\K)) \times \{\inp\}$ and 
    $g_{\ee}(\K_\out) = g(\msg(\K)) \times \{\out\}$.
    Since $\K$ is not empty, neither $\K_\inp$ nor $\K_\out$ are.
    Hence, $\{g_{\ee}(\K_\inp),g_{\ee}(\K_\out)\}$ is a partition of $g_{\ee}(\K)$. 
    Moreover, since $\{\K_i \mid i \in [n]\}$ is a partition of $\K$, we have that $\bigcup_{i\in [n]} g_{\ee}    
    (\K_i)= g_{\ee}(\K)$, and $g_{\ee}(\K_i) \neq \emptyset$, for all $i \in [n]$. 
    What remains to be shown is that the $g_{\ee}(\K_i)$, $i\in [n]$, are pairwise disjoint. 
    So, assume that $e,e' \in \K$ are such that $g_{\ee}    (e) = g_{\ee}(e')$. 
    Thus either $e$ and $e'$ are both output events or they are both input events. 
    Moreover, $g(\msg(e)) = g(\msg(e'))$. 
    Since $g$ has the channel property, the latter implies that $\Ch(\msg(e))=\Ch(\msg(e'))$.
    In other words, $\proc(\msg(e),\out) = \proc(\msg(e'),\out)$ and $\proc(\msg(e),\inp) = \proc(\msg(e'),\inp)$.
    Consequently, $\proc(e) = \proc(e')$. 
%
%    Hence, there exist no $i,j \in [n]$ with $i \neq j$ and $e \in \K_i$ and $e' \in K_j$, such that $g_{\ee}(e) = g_{\ee}(e')$. 
% 
    From this it follows that the $g_{\ee}(\K_i)$, $i \in [n]$ are pairwise disjoint. 
\\
    We now turn to the complement property CP of $g_{\ee}$. 
    Assume $g_{\ee}(e)= g_{\ee}(e')$, for some $e,e' \in \K$.
    By the definition of $g_{\ee}$, either $e$ and $e'$ are both output events or they are both input events. 
    We set $e = (\msg(e),\x) $ and  $e' = (\msg(e'),\x) $ where $\x \in \{\inp, \out\}$. 
    Moreover, $g(\msg(e)) = g(\msg(e'))$. 
    It follows that $g_{\ee}(\cp (e)) = (g(\msg(\cp(e))),\overline{\x}) = (g(\msg(e)),\overline{\x})$. 
    Similarly, $g_{\ee}(\cp(e')) = (g(\msg(\cp(e'))),\overline{x}) = (g(\msg(e')),\overline{x})$. 
    Since $g(\msg(e)) = g(\msg(e'))$, we conclude that $g_{\ee}(\cp(e)) = g_{\ee} (\cp(e'))$, as required. 
\qed
\end{proof}

\begin{example}
\label{Ex-ActFn2} 
(Ex.~\ref{Ex-OTC02} \ctd) 
Consider a record of exchanges in the OTC settlement process in the form of 
message log 
%$\Ka^2=(\K^2,\Att^2)$, 
$\Ka'=(\K',\Att')$ 
with 
%$\K^2=\K^1 \cup \{(3,\out),(3,\inp),(4,\out),$ $(4,\inp) \}$, 
$\K'= \{(\m_i,\out),(\m_i,\inp) \mid i \in [4]\}$, 
$\Att' =(\langle \Cs' \rangle,$ $ \proc',$ $\tm', \langle \tp' \rangle )$. 
We assume that there are two cases $T1$ and $T2$ such that $\Cs'(\m_1)=\Cs'(\m_2)=T1$ and $\Cs'(\m_3)=\Cs'(\m_4)=T2$. 
Messages $m_1$, $m_3$ are Settlement Instructions with   $\tp'(\m_1)= \tp'(\m_3) = SI$, while messages $m_2$, $m_4$ are Settlement Confirmations with $\tp'(\m_2) = \tp'(\m_4) =SC$. 
\\
Message events $(\m_1,\out),(\m_2,\inp)$ and $(\m_3,\out),(\m_4,\inp)$ belong to process $\Inv$ and 
message events $(\m_1,\inp),(\m_2,\out)$ and $(\m_3,\inp),(\m_4,\out)$ to process $\Cus$, see Figure~\ref{Pex09}. 
Hence, we have the following send/receive pairs: 
$\Ch(\m_1) =  \Ch(\m_3) = (\Inv, \Cus)$ and 
$\Ch(\m_2) =  \Ch(\m_4) = (\Cus, \Inv)$. 

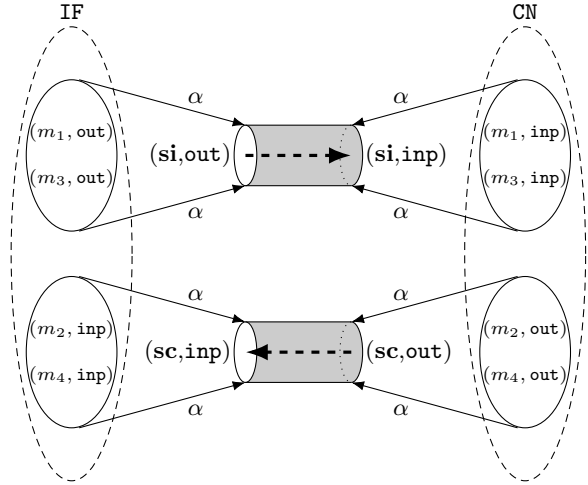
\begin{figure}[!htp]

\centering

\begin{tikzpicture}[event/.style={draw,fill,circle,inner sep=0pt,minimum size=0.8mm},
farrow/.style={-Latex},
         process/.style={densely dashed}  ]
%% x-as links=1 rechts=6 midden=3.5
%% y-as midden-3.5

%% proces links
\node at (0.5,6.7) {\texttt{IF}} ;
\draw[process] (0.5,3.5) ellipse (0.8 and 3);
%% proces rechts
\node at (6.5,6.7) {\texttt{CN}} ;
\draw[process] (6.5,3.5) ellipse (0.8 and 3);
%% BOVEN
\begin{scope}[shift={(0,+1.3)}]
  %% links
  \draw (0.5,3.5) ellipse (0.6 and 1);
  \node at (0.5,3.5+0.30) {\scriptsize $(m_1,\out)$} ;
  \node at (0.5,3.5-0.30) {\scriptsize $(m_3,\out)$} ;
  %% rechts
  \draw (6.5,3.5) ellipse (0.6 and 1);
  \node at (6.5,3.5+0.30) {\scriptsize $(m_1,\inp)$} ;
  \node at (6.5,3.5-0.30) {\scriptsize $(m_3,\inp)$} ;
  %% cylinder pijl binnen
  \coordinate[label={[label distance=2]left:(\textbf{si},\out)}] (sa) at (3.5-0.7,3.5) ;
  \coordinate[label={[label distance=2]right:(\textbf{si},\inp)}] (ra) at (3.5+0.7,3.5) ;
  \draw[very thick,dashed,-Latex] (sa) to (ra) ;
  %% cylinder tekenen
  \draw (3.5-0.7,3.5) ellipse (0.15 and 0.4);
  \draw (3.5+0.7,3.5-0.4) arc (-90:90:0.15 and 0.4);
  \draw[dotted] (3.5+0.7,3.5+0.4) arc (90:270:0.15 and 0.4);
  \draw (3.5-0.7,3.5+0.4) -- (3.5+0.7,3.5+0.4) ;
  \draw (3.5-0.7,3.5-0.4) -- (3.5+0.7,3.5-0.4) ;
  \fill [opacity=0.2] (3.5-0.7,3.5+0.4) -- (3.5+0.7,3.5+0.4)  arc (90:-90:0.15 and 0.4) -- (3.5-0.7,3.5-0.4)  arc (-90:90:0.15 and 0.4) ;
  %% met de hand aangrijpingspunt ellips aangepast
  % linker 
  \draw[farrow] (0.5+0.1, 3.5+1) to node[above,pos=0.7] {$\alpha$} (3.5-0.7,3.5+0.4);
  \draw[farrow] (0.5+0.1, 3.5-1) to node[below,pos=0.7] {$\alpha$} (3.5-0.7,3.5-0.4);
  % rechter ellips
  \draw[farrow] (6.5-0.1, 3.5+1) to node[above,pos=0.7] {$\alpha$} (3.5+0.7,3.5+0.4);
  \draw[farrow] (6.5-0.1, 3.5-1) to node[below,pos=0.7] {$\alpha$} (3.5+0.7,3.5-0.4);
\end{scope}
%% ONDER
\begin{scope}[shift={(0,-1.3)}]
  %% links
  \draw (0.5,3.5) ellipse (0.6 and 1);
  \node at (0.5,3.5+0.30) {\scriptsize $(m_2,\inp)$} ;
  \node at (0.5,3.5-0.30) {\scriptsize $(m_4,\inp)$} ;
  %% rechts
  \draw (6.5,3.5) ellipse (0.6 and 1);
  \node at (6.5,3.5+0.30) {\scriptsize $(m_2,\out)$} ;
  \node at (6.5,3.5-0.30) {\scriptsize $(m_4,\out)$} ;
  %% cylinder pijl binnen
  \coordinate[label={[label distance=2]left:(\textbf{sc},\inp)}] (sa) at (3.5-0.7,3.5) ;
  \coordinate[label={[label distance=2]right:(\textbf{sc},\out)}] (ra) at (3.5+0.7,3.5) ;
  \draw[very thick,dashed,-Latex] (ra) to (sa) ;
  %% cylinder tekenen
  \draw (3.5-0.7,3.5) ellipse (0.15 and 0.4);
  \draw (3.5+0.7,3.5-0.4) arc (-90:90:0.15 and 0.4);
  \draw[dotted] (3.5+0.7,3.5+0.4) arc (90:270:0.15 and 0.4);
  \draw (3.5-0.7,3.5+0.4) -- (3.5+0.7,3.5+0.4) ;
  \draw (3.5-0.7,3.5-0.4) -- (3.5+0.7,3.5-0.4) ;
  \fill [opacity=0.2] (3.5-0.7,3.5+0.4) -- (3.5+0.7,3.5+0.4)  arc (90:-90:0.15 and 0.4) -- (3.5-0.7,3.5-0.4)  arc (-90:90:0.15 and 0.4) ;
  %% met de hand aangrijpingspunt ellips aangepast
  % linker 
  \draw[farrow] (0.5+0.1, 3.5+1) to node[above,pos=0.7] {$\alpha$} (3.5-0.7,3.5+0.4);
  \draw[farrow] (0.5+0.1, 3.5-1) to node[below,pos=0.7] {$\alpha$} (3.5-0.7,3.5-0.4);
  % rechter ellips
  \draw[farrow] (6.5-0.1, 3.5+1) to node[above,pos=0.7] {$\alpha$} (3.5+0.7,3.5+0.4);
  \draw[farrow] (6.5-0.1, 3.5-1) to node[below,pos=0.7] {$\alpha$} (3.5+0.7,3.5-0.4);
\end{scope}
\end{tikzpicture}
\label{PActFn2b}
%\end{subfigure}
\caption{$\alpha = g_{\ee}$ is an action function.}
\label{Pex09}
\end{figure}

\noindent
Let $g : \msg(\K') \rightarrow \{\mbox{\bf{si}}, \mbox{\bf{sc}}\}$ be a message labelling function, where $\mbox{\bf{si}}$ and $\mbox{\bf{sc}}$ are message labels. 
It is defined by 
$g(\m_1) = g(\m_3) = \mbox{\bf{si}}$ and 
$g(\m_2) = g(\m_4) = \mbox{\bf{sc}}$. 
Hence $g$ has the channel property and $g_{\ee}$, the event extension of $g$, is an action function, as can be seen in Figure~\ref{Pex09}. 
\\
Indeed, 
$g_{\ee}$ satisfies (DP) because 
$g_{\ee}\{(\m_1,\out),(\m_3,\out),(\m_2,\inp),(\m_4,\inp)\}$ $=$ 
$\{(\mbox{\textbf{si}},\out), (\mbox{\textbf{sc}},\inp)\}$ is disjoint with  
$g_{\ee}\{(\m_1,\inp),(\m_3,\inp),(\m_2,\out),(\m_4,\out)\}$ $=$  
$\{(\mbox{\textbf{si}},\inp), (\mbox{\textbf{sc}},\out)\}$. 
\\
Furthermore,  
$g_{\ee}\{(\m_1,\out),(\m_2,\out),(\m_3,\inp),(\m_4,\out)\}$ $=$ 
$\{(\mbox{\textbf{si}},\out), (\mbox{\textbf{sc}},\out)\}$ 
and 
$g_{\ee}\{(\m_1,\inp),(\m_2,\inp),(\m_3,\inp),(\m_4,\inp)\}$ $=$  
$\{(\mbox{\textbf{si}},\inp), (\mbox{\textbf{sc}},\inp)\}$ are disjoint sets. 
Hence, $g_{\ee}$ also satisfies (CP). 
\qed
\end{example}

Returning to the definition of action functions (Definition~\ref{Dafngen}), we demonstrate how the complement property CP makes it possible to associate a complement function with each action function. 
\begin{definition}\label{Def-CPf}
    Let $f: \K \rightarrow D$ be a function that satisfies CP. 
    Then the \emph{complement by} $f$ is the function 
    $\cp_f: f(\K) \rightarrow f(\K)$, defined by $\cp_f(f(e)) = f(\cp(e))$ for all $e \in \K$. 
\qed
\end{definition}

%By Lemma~\ref{Lem-CP}, 

That $\cp_f$ is indeed a function whenever $f$ satisfies CP, follows from the observation that whenever $e,e' \in \K$ are such that $f(e) = f(e')$, then $f(\cp(e)) = f(\cp(e'))$. 

\begin{example}\label{Ex-PP1} 
Consider the message log $\Ka = (\K,\Att)$ such that $(\m_1,\out)$, $(\m_2,\out)$, $(\m_1,\inp)$, $(\m_2,\inp) \in \K$. 
Output events $(\m_1,\out)$ and $(\m_2,\out)$ belong to the same process, say $P_1$.
Hence $(\m_1,\inp)$ and $(\m_2,\inp)$ also belong to a common process $P_2$ and $P_2 \neq P_1$. 
Let $\alpha$ be an action function such that $\alpha((\m_1,\out))=\alpha((\m_2,\out))=a$ and $\alpha((\m_1,\inp))=\alpha((\m_2,\inp))=b$. 
Then $\cp_{\alpha}(a)= \alpha(\cp((\m_1,\out))) = \alpha((\m_1,\inp)) = b$. 
Similarly, $\cp_{\alpha}(b)=a$.
\qed
\end{example}

By the next lemma, $\cp_f$ is a complement function if $f$ is an action function. 

\begin{lemma}\label{Lem-cp=comp}
    Let $f: \K \rightarrow D$ be a function that satisfies CP and such that $\{f(\K_\out),f(\K_\inp)\}$ is a partition of $f(\K)$. 
    Then $\cp_f$ is a complement function. 
\end{lemma}

\begin{proof}
It follows from the definitions of $\cp$ and $\cp_f$ that $\cp_f(f(\K_\out)) = f(\cp(\K_\out)) = f(\K_\inp)$ and $\cp_f(f(\K_\inp)) = f(\cp(\K_\inp)) = f(\K_\out)$.  
Moreover, $\cp_f(\cp_f(f(e))) = \cp_f (f(\cp(e))) = f(\cp(\cp(e)) = f(e) $ for all $e \in \K$. 
Thus $\cp_f$ is a complement function. 
\qed   
\end{proof}

The following theorem shows how to associate a distributed communicating alphabet with $\Ka$. 

\begin{theorem}
\label{Thm-DCA}
     Let 
     $\alpha: \K \rightarrow D$ be an action function such that $\alpha(\K)$ is a finite set.
     Let $\tau: \alpha(\K) \rightarrow {\mathcal M}$ be a function that assigns message types to labelled message events in such a way that $\tau(\alpha(\cp(e))) =  \tau(\alpha(e))$ for all $e \in \K$.     
    Then $$\DAx(\Ka,\alpha,\tau) = ([\alpha(\K_1), \ldots, \alpha(\K_n)],[\emptyset,\alpha(\K_{\out}),\alpha(\K_{\inp})], \tau, \cp_{\alpha})$$ is an $n$-\DA{}. 
\end{theorem}

\begin{proof} 
We verify the properties of an $n$-\DA{} as listed in Definition~\ref{Def-DCA}. 
\\
(1) 
Since $\alpha$, being an action function, satisfies the disjointness property (DP) and, moreover, $\alpha(\K)$ is finite, we know that $\alpha(\K_1), \ldots, \alpha(\K_n)$ are non-empty, pairwise disjoint and finite sets with $\alpha(\K) = \bigcup_{i \in [n]} \alpha(\K_i)$.  
\\
(2) 
By the same arguments as above, $\alpha(\K_\inp)$ and $\alpha(\K_\out)$ are pairwise disjoint, finite sets
such that $\emptyset \cup \Sigma_\inp \cup \Sigma_\out = \alpha(\K)$.
\\
(3)
$\bigcup_{i \in [n]} \alpha(\K_i) = \alpha(\K) = \emptyset \cup \Sigma_\inp \cup \Sigma_\out$. 
\\
(4) 
Function $\tau$ assigns a type to all elements of $\alpha(\K)$. 
Note that whenever $e,e' \in \K$ are such that $\alpha(e) = \alpha(e')$, then $\alpha(\cp(e)) = \alpha(\cp(e'))$, because $\alpha$ satisfies CP. 
Hence, the condition that $\tau(\alpha(\cp(e))) =  \tau(\alpha(e))$ for all $e \in \K$, does not lead to inconsistencies.  
\\
(5) 
We note that $n \geq 2$.  
\\
(5.1) 
From Lemma~\ref{Lem-cp=comp} and its proof, we have that $\cp_\alpha$ is a complement function with $\cp_{\alpha}(\alpha(\K_\out)) = \alpha(\K_\inp)$ and $\cp_{\alpha}(\alpha(\K_\inp)) = \alpha(\K_\out)$. 
\\
(5.2) That 
$\cp_\alpha(\alpha(\K_i)) \cap \alpha(\K_i) = \emptyset$ for all $i \in [n]$ can be seen as follows. 
Assume, to the contrary, that $\cp_\alpha(\alpha(d)) = \alpha (e) $ with $d,e \in \K_i$ and $i \in [n]$. 
Hence, $\alpha(\cp(d)) = \alpha(e)$. 
By the definition of process function (\cf Definition~\ref{N07msgB0-2}(1)), 
$d \in \K_i$ implies that $\cp(d) \not\in \K_i$. 
Since $\alpha$ satisfies DP this implies that 
$\alpha(\cp(d)) \neq \alpha(e)$, a contradiction.
\\
(5.3) Finally, by the definition of $\cp_{\alpha}$, 
$\tau(\cp_\alpha(\alpha(e)))= \tau(\alpha(\cp(e))) = \tau(\alpha(e))$ for all $e \in \K$. 
\qed   
\end{proof}

Note that in the statement of Theorem~\ref{Thm-DCA} we use the message types from $\mathcal{M}$ also as types for actions.
Moreover, the condition on $\tau$ is a minimal requirement in view of Definition~\ref{Def-DCA}(5.3). 

From Theorem~\ref{Thm-Seq=finpref}, we know that $\Ev(\Ka)$ is a 
finitary, prefix-closed language. 
Since $\alpha$ is an action function and thus maps every occurrence of a message event to a single action symbol, we have that for every word $w \in \K^*$, $\alpha(w)(i) = \alpha(w(i))$ for all $i\in\dom(w)$.  
Consequently, when $\alpha$ is applied to $\Ev(\Ka)$ the resulting language is also finitary and prefix-closed. 

\begin{corollary}\label{Cor-LabPrefC}
 Let $\alpha: \msg(\K) \rightarrow D$ be an action function.
 Then $\alpha(\Ev(\Ka))$ is a finitary, prefix-closed language.
\qed
\end{corollary}

So, we have a method to translate a message log into a finitary, prefix-closed language, \ie a \emph{(global) event log}, over a distributed communicating alphabet. 
In the next section, we investigate the structure of the labelled sequences comprising this event log.
%in relation to the poset $\po_\K = (\K, \preceq)$ defined by $\Ka=(\K,\Att)$.

%%%%%%%%%%%%%%%%%%%%%%%%%%%%%%%%%%%%%%%%%%%%%%%%%%%%%
%%%%%%%%%%%%%%%%%%%%%%%%%%%%%%%%%%%%%%%%%%%%%%%%%%%%%

\section{The action language/event log of a message log}
\label{Sect-ML} 

%%%%%%%%%%%%%%%%%%%%%%%%%%%%%%%%%%%%%%%%%%%%%%%%%%%%%

%
It is convenient to fix here, for this section and the next, an action function, a type function, and an $n$-\DA{}, all as specified in the statement of Theorem~\ref{Thm-DCA}. 

\begin{center}
\begin{fminipage}{4.0in}
\noindent
\emph{$\alpha: \K \rightarrow D$ is an action function such that $\alpha(\K)$ is a finite set;
\\
$\tau: \alpha(\K) \rightarrow {\mathcal M}$ is such that $\tau(\alpha(\cp(e))) =  \tau(\alpha(e))$ for all $e \in \K$; 
\\ 
and $\DAx(\Ka,\alpha,\tau)$ is the $n$-\DA{} defined by $\Ka$, $\alpha$, and $\tau$. 
}
\end{fminipage}
\end{center}

As we will argue next, the causal structure of $\Ka$ as captured by $\po_\K = (\K, \preceq)$ is preserved in $\alpha(\Ev(\Ka))$. 

We start out by describing how in the words of $\alpha(\Ev(\Ka))$, occurrences of input actions (elements from $\alpha(\K_{\inp})$) can be related to occurrences of their complementary output actions (elements from $\alpha(\K_{\out})$). 
Let us recall that the causality described by the event relation $R_\Ka$ of $\Ka$ (\cf Definition~\ref{D07crl}) requires that in every conversation, each input event is preceded by its corresponding output event. 
This leads, for their labelled occurrences in the words of $\alpha(\Ev(\Ka))$, to the following observation.

\begin{lemma}\label{Lem-wordprefp}
    Let $w \in \alpha(\Ev(\Ka))$, $a \in \alpha(\K_\out)$. 
    Then $\#_{a} (w) \geq \#_{\cp_\alpha(a)} (w)$.
\end{lemma}

\begin{proof}
Let $u \in \Ev(\Ka)$ be such that $w = \alpha(u)$. 
If $\cp_\alpha(a)$ does not occur in $w$, \ie $\cp_\alpha(a) \not\in \ab(w)$, then there is nothing to prove. 
So, let now $i \in \dom(w)$ be such that $w(i) = \cp_\alpha(a)$. 
Hence $u(i) = (\m,\inp)$ for some message $\m$. 
Since $\Ka$ is a message log, $(\m,\inp) \in \Ka$ implies that also $(\m,\out) \in \Ka$ (\cf Definition~\ref{D07msgB0}). 
Moreover, $(\m,\out) \preceq (\m,\inp)$ (\cf Definition~\ref{D07crl}). 

Let $C$ be the unique finite conversation in $\K$ such that $u \in \wdspo(\po_C)$ (\cf Definition~\ref{Def-EL} and Lemma~\ref{Lem-uniqueC}). 
Since $\Cs((\m,\inp)) = \Cs((\m,\out))$ and $C$ is downward closed, it follows that $(\m,\out) \in C$ and hence $(\m,\out)\in \ab(u)$. 
Let $j \in [|u|]$ be such that $u(j) = (\m,\out)$. 
Hence $u(j) \in \dc(u(i))$ and, by Lemma~\ref{rem-pastlo}, $j < i$. 
In other words $(\m,\out)$ precedes $(\m,\inp)$ in $u$. 

Moving back to $w = \alpha(u)$, we conclude that every occurrence of $\cp_\alpha(a)$ as the image under $\alpha$ of a unique input event $(\m, \inp)$ has its own, unique, corresponding occurrence of $\alpha((\m,\out)) = a$ that precedes it in $w$, which proves the statement. 
    \qed
\end{proof}

By Corollary~\ref{Cor-LabPrefC}, the language $\alpha(\Ev(\Ka))$ is prefix-closed. 
Hence it follows from Lemma~\ref{Lem-wordprefp} that every word in $\alpha(\Ev(\Ka))$ has the prefix property with respect to $\cp_\alpha |_{\alpha(\K_{\out})}$.
We identify this function as follows. 

\begin{center}
\begin{fminipage}{4.0in}
\noindent
\emph{
$\varphi: \alpha(\K_{\out}) \rightarrow \alpha(\K_{\inp})$, the function defined by 
$\varphi = \cp_\alpha |_{\alpha(\K_{\out})}$, 
\\
is the restriction of the complement function $\cp_\alpha$ to $\alpha(\K_{\out})$. 
}
\end{fminipage}
\end{center}

Note that $\varphi: \alpha(\K_{\out}) \rightarrow \alpha(\K_{\inp})$ is a bijection. 

\begin{theorem}
\label{Thm-Logprefp}  
$\alpha(\Ev(\Ka))$ has the prefix property with respect to $\varphi$. 
\qed
\end{theorem}

In what follows, we may omit the reference to $\varphi$. 

\begin{example}
\label{Ex-PP}
(Ex.~\ref{Ex-PP1} \ctd)
Assume that the message events $(\m_1,\out)$, $(\m_2,\out)$, $(\m_1,\inp)$, and $(\m_2,\inp)$ all belong to the same case. 
Furthermore, $(\m_1,\out)$ precedes $(\m_2,\out)$ in $P_1$ and 
$(\m_1,\inp)$ precedes $(\m_2,\inp)$ in $P_2$. 
\\
Let $C = \{(\m_1,\out), (\m_1,\inp), (\m_2,\out), (\m_2,\inp)\}$.  
This $C$ is a conversation in $\Ka$ 
with $u=(\m_1,\out)(\m_2,\out)(\m_1,\inp)(\m_2,\inp) \in \wdspo(\po_C)$ and thus $u \in \Ev(\Ka)$, see Figure~\ref{Fig-PP1}(a). 

\begin{figure}
\begin{subfigure}{0.35\textwidth}
\centering
    \includegraphics[width=\linewidth]{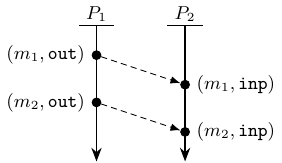}
\caption{Conversation $C$ in $\Ka$.}
\label{Fig-PP1-a}
\end{subfigure}
\begin{subfigure}{0.65\textwidth}
\centering
    \includegraphics[width=\linewidth]{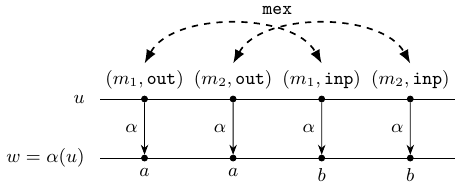}
\caption{$\alpha(u)$ has the prefix property.}
\label{Fig-PP1-b}
\end{subfigure}

\caption{}
\label{Fig-PP1}
\end{figure}

\noindent
Then $\alpha(u)=aabb$ as shown in Figure~\ref{Fig-PP1}(b). 
Clearly, $\#_{b} (v) \leq \#_{a} (v)$ for all prefixes $v$ of $\alpha(u)$.
Hence, $w = \alpha(u)$ has the prefix property. 
\qed
\end{example}

As can be seen in Example~\ref{Ex-PP}, there is a clear relation between the occurrences of complementary message events in an event sequence, which is lacking in the labelled version. 
Assignment functions (introduced in~\cite{DBLP:journals/topnoc/KwantesK22}) explicitly relate the individual occurrences of input and output actions. 

\begin{definition}\label{Def-AF}
Let $w \in \alpha(\K)^*$ be a word with the prefix property. 
An \emph{assignment function} (with respect to $\varphi$) for $w$ 
is an injective function $\theta : \occ(w) \cap (\alpha(\K_\inp) \times \mathbb{N}) \rightarrow \occ(w) \cap (\alpha(\K_\out) \times \mathbb{N})$
with, for all $(b,j) \in \occ(w)$ where $b \in \alpha(\K_{\inp})$: 
$\theta((b,j))=(\cp_\alpha(b),i) \in \occ(w)$ 
for some $i$ such that $\pos_w((\cp_\alpha(b),i))$ $<$ $\pos_w((b,j))$. 
\qed
\end{definition}

Assignment functions capture abstract channel policies as they define output/input relations through pairs of occurrences of complementary output and input actions.

\begin{example}
\label{Ex-ASF} 
(Ex.~\ref{Ex-PP} \ctd) 
Consider $\occ(w)=\{(a,1),(a,2),$ $ (b,1),(b,2)\}$. 
\\
Function $\theta_1$ on $\occ(w)$ is defined by $\theta_1((b,1))=(a,2)$ and $\theta_1((b,2))=(a,1)$. 
This function is an assignment function for $w$. 
As illustrated in Figure~\ref{Fig-ASF}(a), it  corresponds to a last-in-first-out (LIFO) channel policy.  

\begin{figure}
\begin{subfigure}{0.48\textwidth}
\centering
    \includegraphics[width=\linewidth]{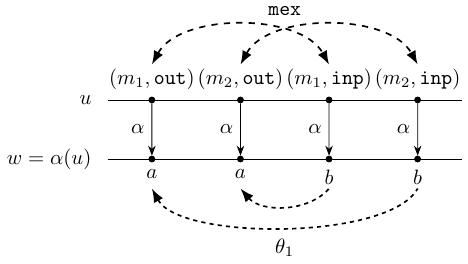}
\caption{$\theta_1$}\label{fig:sub:geschaald-A}
\end{subfigure}
\begin{subfigure}{0.48\textwidth}
\centering
    \includegraphics[width=\linewidth]{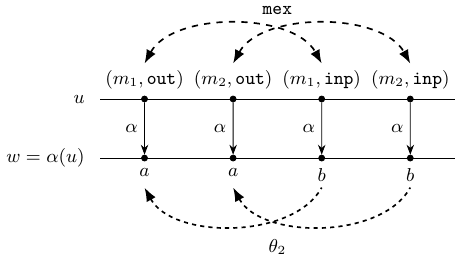}
\caption{$\theta_2 = \theta_{u,w,C}$}\label{fig:sub:geschaald-B}
\end{subfigure}
    \caption{}    %\textcolor{red}{is geschaalde pdf.}
    \label{Fig-ASF}
\end{figure}

\noindent
On the other hand, function $\theta_2$ on $\occ(w)$, defined by $\theta_2((b,1))=(a,1)$ and $\theta_2((b,2))=(a,2)$, specifies a first-in-first-out (FIFO) channel policy, see Figure~\ref{Fig-ASF}(b).
Another feature of $\theta_2$ is that it \emph{preserves messages} in the sense that $(b,1)$ and $\theta_2((b,1))$ have an underlying message event with the same message $\m_1$ and, similarly, $(b,2)$ and $\theta_2((b,2))$ have an underlying message event with the same message, in this case $\m_2$.
\\
Formally, 
$u(\pos_w((b,1))) = (\m_1,\inp)$ and  $u(\pos_w(\theta_2((b,1)))) = (\m_1,\out)$ as well as 
$u(\pos_w((b,2))) = (\m_2,\inp)$ and  $u(\pos_w(\theta_2((b,2)))) = (\m_2,\out)$. 
\qed
\end{example}

We now consider the labelled event sequences defined by a finite conversation in $\Ka$. 
As we prove in the following lemma, 
all labelled event sequences of a single finite conversation have the same set of occurrences. 
Moreover, action functions are injective, when restricted to the event sequences of a single finite conversation.\footnote{Thus, if two distinct event sequences in $\Ev(\Ka)$ yield the same action sequence, they must belong to different finite conversations in $\Ka$; compare this with Lemma~\ref{Lem-uniqueC} by which no event sequence from $\Ev(\Ka$) can belong to more than one conversation.}
Finally, the following lemma also shows that any finite conversation has essentially only one assignment function that preserves messages (as explained in Example~\ref{Ex-ASF}).

\begin{lemma}\label{Lem-MesPres}
Let $u,u' \in \wdspo(\po_C)$ where $C$ is a finite conversation in $\Ka$. 
Let $w =\alpha(u)$ and $w' = \alpha(u')$.
Then 
\\
(1) $\occ(w) = \occ(w')$. 
\\
(2) If $w = w'$, then $u = u'$. 
\\
(3)
If $(b,j) \in \occ(w)$ is such that $u(\pos_w((b,j))) = (\m,\inp)$ for some message $\m$, then $(\m,\out) \in \ab(u)$ and $\pos_u((\m,\inp)) > \pos_u(\m,\out)$.
\\
(4)
If $(b,j) \in \occ(w)$ is such that $u(\pos_w((b,j))) = (\m,\inp)$ for some message $\m$ and $i$ is such that $u(\pos_w((\cp_\alpha(b),i))) = (\m,\out)$, then also 
\\
$u(\pos_{w'}((b,j))) = (\m,\inp)$ and $u(\pos_{w'}((\cp_\alpha(b),i))) = (\m,\out)$. 
\end{lemma} 

\begin{proof}
    (1) Since both $u$ and $u'$ are sequences defined by $\po_C$, $\ab(u) = C = \ab(u')$.     
    Then $\occ(\alpha(u)) = \occ(\alpha(u'))$ follows, as every element of $C$ occurs exactly once in $u$ and in $u'$. 

    (2) Assume $\alpha(u) = \alpha(u')$, but $u \neq u'$. 
    Let $i \in [|u|]$ be such that $u(i) \neq u'(i)$ and $u(j) = u'(j)$ for all $j \in [i-1]$. 
    Note that all message events that occur in $u$, occur only once. 
    Similarly for those in $u'$.
    Hence $u(i),u'(i) \not\in \{u(1), \ldots u(i-1)\}$. 
    Since $\alpha(u(i)) = \alpha(u'(i))$ and $\alpha$ is an action function, it follows from the disjointness property DP (Definition~\ref{Dafngen}) that $\proc(u(i)) = \proc(u'(i))$. 
    Hence $\tm(u(i))\neq \tm(u'(i))$. 
    Thus, by the definition of $\preceq_C$ (\cf Definitions~\ref{D07crl} and~\ref{Def-Con}), either $u(i) \preceq_C u'(i)$ or $u'(i) \preceq_C u(i)$ in $\po_C$ and so by Lemma~\ref{rem-pastlo}, either $u(i)$ is an element in the past of $u'(i)$ or $u'(i)$ is an element in the past of $u(i)$.
    In other words, $u(i)$ or $u'(i)$ is an element of $\{u(1), \ldots , u(i-1)\}$, a contradiction. 

    (3) This follows from the observation that $C$ is downward-closed in $(\K,\preceq)$ and hence every input event in $u$ is preceded by its corresponding output event. 

    (4) Let $(b,j) \in \occ(w)$ be such that $u(\pos_w((b,j))) = (\m,\inp)$ for some message $\m$. 
    Hence by (2), $(b,j) \in \occ(w')$. 
    By the disjointness property DP of $\alpha$ (Definition~\ref{Dafngen}), 
    all message events $x, x' \in C$ such that $\alpha(x) = \alpha(x')$ belong to the same process and are thus linearly ordered by $\preceq_C$. 
    This implies that $u(\pos_{w'}((b,j))) = (\m,\inp)$.  
    Since $w,w' \in \alpha(\Ev(\Ka))$, 
    we have from (3) that there are $i,i'$ such that $u(\pos_w((\cp_\alpha(b),i))) = (\m,\out)$ and $u(\pos_{w'}((\cp_\alpha(b),i'))) = (\m,\out)$.
    Again, by the disjointness property of $\alpha$, $i = i'$. 
 \qed
\end{proof}

We are now ready to show that any finite conversation $C$ defines a single message preserving assignment function which moreover is shared by all message event sequences from $\wdspo(\po_C)$. 
Note that, by Lemma~\ref{Lem-MesPres}(1), $\occ(w) = \occ(w')$ for all $w \in\alpha(\wdspo(\po_C))$ for every finite conversation $C$ in $\Ka$. 
Therefore, we may write $\occ_\alpha(C)$ to denote the set of all occurrences in any word from $\alpha(\wdspo(\po_C))$.

\begin{definition}
\label{Def-mpASF}
    Let $C$ be a finite conversation in $\Ka$.  
    Let $u \in \wdspo(\po_C)$ and  $w = \alpha(u)$. 
    Then $\theta_{u,w,C} : \occ(w) \cap (\alpha(\K_\inp) \times \mathbb {N}) \rightarrow \occ(w) \cap (\alpha(\K_\out) \times \mathbb{N})$ is the \emph{$(u,w)$-message preserving} function for $C$ if, for all $(b,j) \in \occ(w)$ with $b \in \alpha(K_\inp)$, $\theta_{u,w,C}((b,j)) = (\cp_\alpha(b),i)$ provided $u(\pos_w((b,j))) = (\m,\inp)$ and $i$ is such that $u(\pos_w((\cp_\alpha(b),i))) = (\m,\out)$.
\qed   
\end{definition}

Let $C$ be a finite conversation in $\Ka$.
Then, by Lemma~\ref{Lem-MesPres}(4), $\theta_{u,w,C} = \theta_{u',w',C}$ for all $u, u' \in \wdspo(\po_C)$, $w = \alpha(u)$, and $w' = \alpha(u')$. 
We refer to this function as the \emph{message preserving function for} $C$ and denote it by $\theta_C$.

\begin{example}
\label{Ex-asfCon} 
(Ex.~\ref{Ex-ASF} \ctd) 
Consider again $w = \alpha(u) = aabb$ with sequence $u=(\m_1,\out)(\m_2,\out)(\m_1,\inp)(\m_2,\inp) \in \wdspo(\po_C)$, \cf Figure~\ref{Fig-ASF}(b).  
Then function $\theta_{u,w,C}$ is defined by $\theta_{u,w,C}((b,1))=(a,1)$ and $\theta_{u,w,C}((b,2))=(a,2)$. 
Hence $\theta_{u,w,C} = \theta_2$.
\\
Next, let $u' =(\m_1,\out)(\m_1,\inp)(\m_2,\out)(\m_2,\inp)$. 
As can be seen in Figure~\ref{Fig-PP1}(a), also $u'\in \wdspo(\po_C)$ 
With $w' = \alpha(u') = abab$, we then have $\theta_{u',w',C}(b,2)=(a,2)$ and $\theta_{u',w',C}(b,1)=(a,1)$, see Figure~\ref{Fig-ASF3}.  
\\
Hence $\theta_{u',w',C}=\theta_{u,w,C} = \theta_C$.
\qed

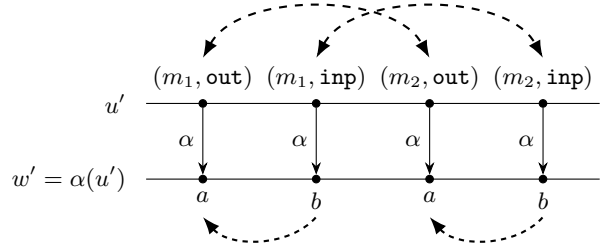
\begin{figure}[!htp]
\centering
\begin{tikzpicture}[xscale=1.5]
%% boven
\draw[mexmap] (1,1.6) to[out=90-20,in=90+20, looseness=1.2] (3,1.6) ;
\draw[mexmap] (2,1.6) to[out=90-20,in=90+20, looseness=1.2] (4,1.6) ;
%\node at (2.5,2.5) {\cp} ;
%% diagram
\draw (0.5,1) -- (4.5,1) ;
\draw (0.5,0) -- (4.5,0) ;
\node [left] at (0.4,1) {$u'$} ;
\node [left] at (0.4,0) {$w'=\alpha(u')$} ;
    \foreach \x/\y [count=\xx] in 
          { {m_1,\out}/a,{m_1,\inp}/b,{m_2,\out}/a,{m_2,\inp}/b }
    {
    \node[punt,label=above:{\small $(\x)$}] (A) at (\xx,1) {} ;
    \node[punt,label=below:{$\y$}] (B) at (\xx,0) {} ;    
    \draw[alphamap] (A) edge node[left] {$\alpha$} (B)  ;
    }
%% onder
\draw[thetamap] (1,0-0.5) to[out=-90+30,in=-90-30, looseness=1.2] (2,0-0.5) ;
\draw[thetamap] (3,0-0.5) to[out=-90+30,in=-90-30, looseness=1.2] (4,0-0.5) ;
%\node at (2.5,-1.2) {$\theta_2$} ;
\end{tikzpicture}

\caption{$\theta_{u',w',C}$}
\label{Fig-ASF3}
\end{figure}

\end{example}

\begin{lemma}\label{Lem-AssMess}
    Let $C$ be a finite conversation in $\Ka$ and $w\in\alpha(\wdspo(\po_C))$. 
    Then $\theta_C$ is an assignment function for $w$.
\end{lemma}

\begin{proof}
    Firstly, we note that $w$ has the prefix property (see Theorem~\ref{Thm-Logprefp}). 
\\
    Let $u \in \wdspo(\po_C)$ be such that $\alpha(u) = w$ and consider $(b,j) \in \occ(w)$ with $u(\pos_w((b,i))) = (\m,\inp)$ for some message $\m$. 
    Recall that $\theta_C = \theta_{u,w,C}$.
   Thus, by Definition~\ref{Def-mpASF}, we have $\theta_C((b,j))=(\cp_\alpha(b),i) \in \occ(w)$ for some $i$ such that $u(\pos_w((\cp_\alpha(b),i))) = (\m,\out)$.
    By Lemma~\ref{Lem-MesPres}(3), indeed $(\m,\out) \in \ab(u)$ and $\pos_u((\m,\out)) < \pos_u((\m,\inp))$.
    Since $\pos_u((\m,\inp)) = \pos_w((b,j))$ and $\pos_u((\m,\out)) = \pos_w((\cp_\alpha(b),i))$, we have $\pos_w((\cp_\alpha(b),i)) < \pos_w((b,j))$. 
\\
    So, according to Definition~\ref{Def-AF}, it follows that $\theta_C$ is an assignment function of $w$ once we have established that $\theta_C$ injective. 
    So, suppose that $(b,j), (c,k) \in \occ_\alpha(C)$ with $b,c \in \alpha(K_\inp)$, are such that $\theta_C((b,j)) = \theta_C((c,k))$.
\\    
    We observe that $\theta_C((b,j)) =     
    (\cp_\alpha(b),i) = (\cp_\alpha(c),\ell) = \theta_C((c,k))$, for certain $i,\ell$. 
   Hence, $i = \ell$ and $\cp_\alpha(b) = \cp_\alpha(c)$. 
   The latter implies, by the complement property CP of $\alpha$ (\cf Definition~\ref{Dafngen}), that $b = c$. 
    Let $u(\pos_w(\cp_\alpha(b),i)) = (\m,\out)$.     
    It follows from the definition of $\theta_C$, that $u(\pos_w((b,j))) = (\m,\inp) = u(\pos_w((c,k)))$, 
    Since $u$ is an event sequence in which every message event occurs at most once, we have $j = k$.
    Consequently, $(b,j) = (c,k)$ which shows that $\theta_C$ is an injective function. 
\qed   
\end{proof}

We are now ready to relate the common causality structure $\po_C = (C,\preceq_C)$ underlying all sequences of a finite conversation $C$ in $\Ka$, to structures associated with assignment functions for sequences in $\alpha(\wdspo(C))$. 
To do this, we use the following definition, originally from~\cite{DBLP:journals/topnoc/KwantesK22}, which describes a  relation on the occurrences in a word derived from an assignment function for that word.  

\begin{definition}\label{Def-ASPO}
Let $w \in \alpha(\Ev(\Ka))$ 
and let $\theta$ be an assignment function for $w$. 
Let 
$(a,j),(a',j') \in \occ(w)$
with $a \in \alpha(\K_k)$ and $a' \in \alpha(\K_\ell)$ for some $k,\ell \in [n]$.
Then $(a,j) \leq_{\theta} (a',j')$ if either 
\\
$(1)$ $k=\ell$ and $\pos_w((a,j)) \le \pos_w((a',j'))$ or
\\
$(2)$ $k \neq \ell$, $a' \in \alpha(\K_{\inp})$ and $\theta(a',j')=(a,j)$.
\qed
\end{definition}

Clearly, we have $\Id_{\occ_\alpha(C)} \subseteq \; \leq_\theta$ 
for all assignment functions $\theta$ for $w \in \alpha(\wdspo(\po_C))$ with $C$ a finite conversation. 
Moreover, by Lemma 5 in~\cite{DBLP:journals/topnoc/KwantesK22}, the transitive closure $\le^+_{\theta}$ of $\le_{\theta}$ is a partial order on $\occ_\alpha(C)$. 

As the next lemma shows, 
the relation $\leq_{\theta_C}$ defined by a message preserving function $\theta_C$ for a finite conversation $C$ in $\Ka$, can be seen as a transfer of $R_\Ka$ from $C$ to $\occ_\alpha(C)$. 

\begin{lemma}\label{Lem-AssPO}
    Let $C$ be a finite conversation in $\Ka$.
    Let $u\in\wdspo(\po_C)$ and $w = \alpha(u)$.
    Let $(a,j),(a',j') \in \occ_\alpha(C)$ 
    be such that $(a,j) \neq (a',j')$.
    Then $(a,j) \leq_{\theta_C} (a',j')$ if and only if 
    $e \, R_\Ka \, e'$ where $e = u(\pos_w((a,j)))$ and $e'= u(\pos_w((a',j')))$. 
\end{lemma}

\begin{proof}
By Lemma~\ref{Lem-AssMess}, $\theta_C$ is an assignment function for $w$.
Note, furthermore, that $a = \alpha(e)$ and $a' = \alpha(e')$.
Let $k$ be such that $e \in \K_k$ 
and let $\ell$ be such that $e'\in \K_\ell$. 
\\
We first consider the case that $k = \ell$.
Then, $\proc(e) = \proc(e')$ which implies that $\tm(e)<\tm(e')$ or $\tm(e')<\tm(e)$.
From Definition~\ref{Def-ASPO}, it follows that  
$(a,j) \leq_{\theta_C} (a',j')$ if and only if 
$\pos_w((a,j)) < \pos_w((a',j'))$. 
Note that 
$\pos_w((a,j)) < \pos_w((a',j'))$ if and only if 
$e = u(\pos_w((a,j)))$ occurs in $u$ before $e' = u(\pos_w((a',j')))$ does. 
Finally, $e$ occurs before $e'$ in $u$ if and only if $\tm(e) < \tm(e')$ if and only if $e \; R_\Ka \; e'$. 
\\
Next we assume that $k \neq \ell$. 
Then we know from Definition~\ref{Def-ASPO}, that $(a,j) \leq_{\theta_C} (a',j')$ if and only if 
$a' \in \alpha(\K_{\inp})$ and $\theta_C((a',j'))=(a,j)$. 
By definition,  
$\theta_C((a',j'))=(a,j)$ if and only if 
$e = u(\pos_w((a,j))) = (\m,\inp)$ and 
$e' = u(\pos_w((a',j'))) = (\m,\out)$ for some message $\m$, 
which by Definition~\ref{D07crl} holds 
if and only if $e \;R_\Ka\; e'$.  
\qed
\end{proof}

As an immediate consequence of Lemma~\ref{Lem-AssPO}, we have the result that the partial order $\le^+_{\theta_C}$ is essentially the same as $\preceq_C$.

\begin{theorem}\label{Thm-IsoPO}
    Let $C$ be a finite conversation in $\Ka$ and let $w \in \alpha(\wdspo(\po_C))$.
    Then $(\occ_\alpha(C),\leq^+_{\theta_C})$ and $\po_C = (C,\preceq_C)$ are isomorphic posets. 
    \qed
\end{theorem}

\begin{proof}
Let $u\in\wdspo(\po_C)$ be such that $w = \alpha(u)$.
Define $f: \occ(w) \rightarrow \ab(u)$ by $f(x) = u(\pos_w(x))$ for all $x \in \occ(w)$. 
This function is a bijection that relates every occurrence $x \in \occ(w)$ to its underlying message event $u(\pos_w(x))$. 
Thus, by Lemma~\ref{Lem-AssPO}, the partial orders $\le^+_{\theta_C}$ and $\preceq_C$ are isomorphic. 
\qed
\end{proof}

Theorem~\ref{Thm-IsoPO} specifically refers to the message preserving assignment functions associated with finite conversations. 
The following definition (adapted from \cite{DBLP:journals/topnoc/KwantesK22}) introduces assignment functions based on the first-in-first-out policy as exemplified in Example~\ref{Ex-ASF}. 

\begin{definition}
Let $C$ be a finite conversation in $\Ka$. 
The $\tt{fifo}$-function for $C$, $\theta_{C,\tt{fifo}}: \occ_\alpha(C) \cap (\alpha(\K_\inp) \times \mathbb {N}) \rightarrow \occ_\alpha(C) \cap (\alpha(\K_\out) \times \mathbb{N})$, is defined by 
$\theta_{C,\tt{fifo}}(b,j) = (\cp_\alpha(b),j)$ for all $(b,j) \in \occ(w)$ with $b \in \alpha(K_\inp)$. 
\qed
\end{definition}

\begin{example}
\label{Ex-ASF2} 
(Ex.~\ref{Ex-ASF} \ctd) 
We recall that $w = aabb = \alpha(u)$ and that the assignment function $\theta_2$ is defined by $\theta_2((b,1))=(a,1)$ and $\theta_2((b,2))=(a,2)$. 
As illustrated in Figure~\ref{Fig-ASF}(b) and referred to in Example~\ref{Ex-ASF} this assignment function is the $\tt{fifo}$-function for $C$. 
\qed
\end{example}

Because all sequences defined by finite conversations have the prefix property, it is immediately clear that every $\tt{fifo}$-function is an  assignment function. 

\begin{lemma}\label{Lem-AssFIFO}
    Let $C$ be a finite conversation in $\Ka$ and $w\in\alpha(\wdspo(\po_C))$. 
    Then $\theta_{C,\tt{fifo}}$ is an assignment function for $w$.
\qed
\end{lemma}

As argued in~\cite{DBLP:journals/topnoc/KwantesK22}, $\tt{fifo}$-functions are the least restrictive among the assignment functions. 
Here, we formulate this as follows. 

\begin{lemma}\label{Lem-fifoPO}
    %Let $w \in \alpha(\K)^*$ be a word with the prefix property 
    Let $C$ be a finite conversation in $\Ka$ and $w\in\alpha(\wdspo(\po_C))$.
    Let $\theta$ be an assignment function for $w$. 
    Then $\le^+_{\theta_{C,\fifo}} \; \subseteq \; \le^+_{\theta}$. 
\end{lemma}

\begin{proof}
    We start from Definition~\ref{Def-ASPO}.
    Let $(a,j),(a',j') \in \occ(w)$ with $a \in \alpha(\K_k)$ and $a' \in \alpha(\K_\ell)$ for some $k,\ell \in [n]$,     
    be such that $(a,j) \leq_{\theta_{w,\tt{fifo}}} (a',j')$. 
\\
    If $k=\ell$, then $\pos_w((a,j)) \le \pos_w((a',j'))$ and hence also $(a,j) \leq_{\theta} (a',j')$. 
\\
    Otherwise, $k \neq \ell$ and $a' \in \alpha(\K_{\inp})$ with $\theta_{C,\tt{fifo}}(a',j')=(a,j)$. 
    Thus $j = j'$. 
    Let $\theta(a',j)=(a,i)$ for some $i$ such that $(a,i) \in \occ(w)$. 
    There are two cases to consider. 
\\
    (1) $i \geq j$: 
    In this case, $(a,j) \leq^+_{\theta} (a,i) \leq_{\theta} (a',j)$ and so $(a,j) \leq^+_{\theta} (a',j) = (a',j')$. 
\\
    (2) $i < j$:
    Since $\theta$ is an injective function on $\occ(w)$, there are $j'' < j$ and $i' \geq j$ such that $\theta((a',j'')) = (a,i')$.
    So, $(a,j) \leq^+_\theta (a,i') \leq_\theta(a',j'') \leq^+_\theta (a',j) = (a',j')$. 
\\
    Thus, in all cases, $(a,j) \leq_{\theta_{C,\tt{fifo}}} (a',j')$ implies that $(a,j) \leq^+_\theta (a',j')$.
    Hence, $\le^+_{\theta_{C,\fifo}} \; \subseteq \; \le^+_{\theta}$.     
\qed
\end{proof}

In other words, the partial order on the occurrences of a word $w$ defined by a $\tt{fifo}$-function, is the smallest of all partial orders defined by assignment functions for $w$. 
\\
Consider, \eg the assignment functions $\theta_1$ and $\theta_2$ from the  Examples~\ref{Ex-ASF} and~\ref{Ex-ASF2}. 
Function $\theta_2$ is a $\tt{fifo}$-function. 
Moreover, $(\m_1,\out)$ precedes $(\m_2,\out)$ in $P_1$ and $(\m_1,\inp)$ precedes $(\m_2,\inp)$ in $P_2$.
So, as can be seen in Figure~\ref{Fig-ASF}, the partial order induced by $\theta_2$  on the occurrences of $w$ is less restrictive than the partial order induced by $\theta_1$. 
In particular, $\theta_2$ would allow to change in $w$, the order of the occurrences $(a,2)$ and $(b,2)$. 

This observation then leads to the last part of this section, where we consider the labelled sequences derived from the linear extensions of partial orders $\le^+_{\theta}$ defined by assignment functions $\theta$. 
The following definition is from~\cite{DBLP:journals/topnoc/KwantesK22}. 

\begin{definition}\label{Def-Label}
Let $w \in \alpha(\Ev(\Ka))$ 
and let $\theta$ be an assignment function for $w$.
Then $\lin_\theta(w) = \{v \in \alpha(\K)^* \mid \occ(v) = \occ(w) \mbox{ and } \forall x,y \in \occ(v): x \leq^+_\theta y \Rightarrow \pos_v(x) \leq \pos_v(y)\}$ is the set of \emph{$\theta$-linearisations of} $w$. 
\qed
\end{definition}

Note that Definition~\ref{Def-Label} describes $\theta$-linearisations as sets of words, in contrast to linear extensions of partial orders which are linearly ordered sets. 
As a matter of fact, the $\theta$-linearisations of $w$ are  labelled versions of the sequences in $\wdspo(\occ(w),\leq^+_{\theta_C})$. 

\begin{lemma}\label{Lem-Labellin}
    Let $w \in \alpha(\Ev(\Ka))$ and let $\theta$ be an assignment function for $w$. 
    Then $\lin_\theta(w) = \lab(\wdspo((\occ(w),\leq^+_{\theta})))$,
    where $\lab: \alpha(\K) \times (\mathbb{N} \setminus \{0\}) \rightarrow \alpha(\K)$ is defined by $\lab((a,j)) = a$ for all $a \in \alpha(\K)$ and $j\geq 1$. 
    \qed    
\end{lemma}

Combining Definition~\ref{Def-Label} with Lemma~\ref{Lem-fifoPO} leads immediately to the following observation. 

\begin{lemma}\label{Lem-Lin}
     Let $C$ be a finite conversation in $\Ka$ and $w\in\alpha(\wdspo(\po_C))$.
    Let $\theta$ be an assignment function for $w$. 
    Then $\lin_{\theta}(w) \subseteq \lin_{\theta_{C,\tt{fifo}}}(w)$
\qed
\end{lemma}

\begin{example}
\label{Ex-ASF3} 
Let $\Ka' = (\K,\Att')$ be a message log with message events $(\m_1,\out)$, $(\m_2,\out)$, $(\m_1,\inp)$, $(\m_2,\inp) \in \K$. 
Like in Example~\ref{Ex-PP1}, the output events $(\m_1,\out)$ and $(\m_2,\out)$ belong to a process $P_1$ and the input events $(\m_1,\inp)$ and $(\m_2,\inp)$ to a second process $P_2$. 
Here, however, $(\m_1,\out)$ precedes $(\m_2,\out)$ in $P_1$ and 
$(\m_2,\inp)$ precedes $(\m_1,\inp)$ in $P_2$. 
\\
So, $C = \{(\m_1,\out), (\m_1,\inp), (\m_2,\out), (\m_2,\inp)\}$ is a conversation in $\Ka'$, see Figure~\ref{Fig-PP1'}(a).
Consider $u=(\m_1,\out)(\m_2,\out)(\m_2,\inp)(\m_1,\inp) \in \wdspo(\po_C)$. 

\begin{figure}[!ht]
\begin{subfigure}{0.35\textwidth}
    \centering
    \includegraphics[width=\linewidth]{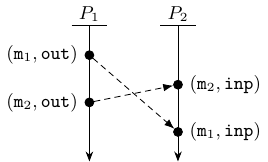}
\caption{Conversation $C$ in $\Ka'$.}
\label{Fig-PP1'-a}
\end{subfigure}
\begin{subfigure}{0.65\textwidth}
\centering
    \includegraphics[width=\linewidth]{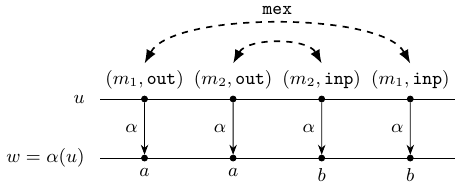}
\caption{$\occ(w) = \{(a,1),(a,2),(b,1),(b,2)\}$. }
\label{Fig-PP1'-b}
\end{subfigure}

\caption{}
\label{Fig-PP1'}
\end{figure}

\noindent
Now, action function $\alpha$ is defined such that  $\alpha(\m_1,\out)=\alpha(\m_2,\out)=a$ and $\alpha(\m_1,\inp)=\alpha(\m_2,\inp)=b$. 
As illustrated in Figure~\ref{Fig-PP1'}(b), we have $\alpha(u)=aabb$ and $\occ(w) = \occ_{\alpha}(C) = \{(a,1),(a,2),(b,1),(b,2)\}$. 
\\
Hence, 
$\theta_C(b,1)=(a,2)$ and $\theta_C(b,2)=(a,1)$ where 
$\theta_C$ is the message preserving assignment function for $C$. 
This function differs from the $\fifo$-function $\theta_{C,\tt{fifo}}$ which is defined by $\theta_{C,\fifo}(b,1)=(a,1)$ and $\theta_{C,\fifo}(b,2)=(a,2)$. \\
The partial orders $\leq^+_{\theta_{C_2}}$ and $\leq^+_{\theta_{C_2,\fifo}}$ are given in Figures~\ref{Fig-PoASF}(b,c). 
There we see that $\leq^+_{\theta_C} \, \subseteq \, \leq^+_{\theta_{C,\fifo}}$. 
\\
In fact, there are no other assignment functions for $w$ than $\theta_C$ and $\theta_{C,\tt{fifo}}$. 

\begin{figure}[!htp]
\begin{subfigure}{0.48\textwidth}
\centering
\begin{tikzpicture}[xscale=1.2]
%% y-coordinaat 4- want tijd omgekeerd
    \node (a2) at (1,4-2) {$(a,2)$} ;
    \node (a1) at (1,4-0.5) {$(a,1)$} ;
    \node (b2) at (3,4-2.8) {$(b,2)$} ;
    \node (b1) at (3,4-1.3) {$(b,1)$} ;
% google-de-google    
\draw[-{Straight Barb[angle'=70]},thick] (a1) edge (a2) edge (b2) edge[dashed] (b1) 
          (a2) edge (b1) edge[dashed] (b2)
          (b1) edge (b2) ;    
%\caption{$\leq^+_{\theta_{C_2}}$}%\label{PConAssFn02c}
\end{tikzpicture}
\caption{$\leq^+_{\theta_{C_2}}$ minus identity}\label{Fig-PoASF-a}
\end{subfigure}
\begin{subfigure}{0.48\textwidth}
\centering
\begin{tikzpicture}[xscale=1.2]
%% y-coordinaat 4- want tijd omgekeerd
    \node (a2) at (1,4-2) {$(a,2)$} ;
    \node (a1) at (1,4-0.5) {$(a,1)$} ;
    \node (b2) at (3,4-2.8) {$(b,2)$} ;
    \node (b1) at (3,4-1.3) {$(b,1)$} ;
% google-de-google    
\draw[-{Straight Barb[angle'=70]},thick] 
          (a1) edge (a2) edge[dashed] (b2) edge (b1) 
          (a2) edge (b2)
          (b1) edge (b2) ;    
%\caption{$\leq^+_{\theta_{C_2,\fifo}}$}%\label{PConAssFn02d}

\end{tikzpicture}
\caption{$\leq^+_{\theta_{C_2,\fifo}}$ minus identity}\label{Fig-PoASF-b}
\end{subfigure}

\caption{$\leq_{\theta_C}$ and $\leq_{\theta_{C,\fifo}}$ are in bold.}
\label{Fig-PoASF}
\end{figure}
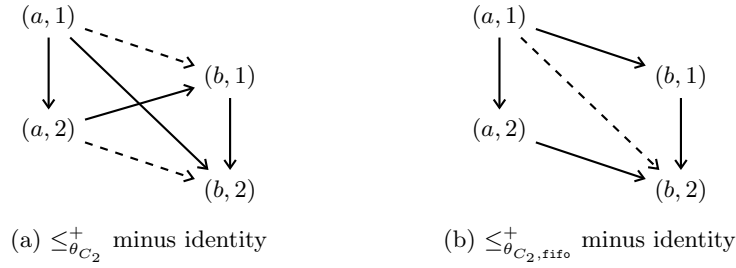

\noindent
The set of $\theta_C$-linearisations of $w = aabb$ comprises only $w$, since $\leq^+_{\theta_C}$ is a total order on $\occ(w)$ which inhibits the swapping of occurrences in $w$. 
The set of $\theta_{C,\tt{fifo}}$-linearisations of $w$ consists of $w$ and $w' = abab$.
In this case, we can swap the second $a$ and the first $b$, as $(a,2)$ and $(b,1)$ are not related by $\leq^+_{\theta_{C,\tt{fifo}}}$. 
\qed
\end{example}

Wrapping up Lemmas~\ref{Lem-Lin},~\ref{Lem-Labellin}, and Theorem~\ref{Thm-IsoPO} yields: 

\begin{theorem}
Let $C$ be a finite conversation in $\Ka$ and let $w \in \alpha(\wdspo(\po_C))$. 
Then,  
$\alpha(\wdspo(\po_C)) = \lin_{\theta_C}(w) \subseteq  
\lin_{\theta_{C,\tt{fifo}}}(w)$. 
\qed
\end{theorem}

Recall that the global event log $\alpha(\Ev(\Ka))$ defined by the message log $\Ka$ is 
$\bigcup\{\alpha(\wdspo(\po_C)) \mid C \text{ a finite conversation in } \Ka \;\}$. 
From the above we may thus conclude that, for every finite conversation $C$ in $\Ka$, the $\theta_C$-linearisations of its sequences as defined by the message preserving assignment function $\theta_C$ are all in $\alpha(\Ev(\Ka))$. 
Moreover, the causality between message events as described by $\po_C$ is preserved by the translation from message log $\Ka$ to $\Ev(\Ka)$ and event log $\alpha(\Ev(\Ka))$. 

\begin{corollary}
\label{Cor-Lin}
$\alpha(\Ev(\Ka)) = $
\\
$\bigcup\{lin_{\theta_C}(w) \mid C \text{ a finite conversation in } \Ka \text{ and } w \in \alpha(\wdspo(\po_C)) \; \} \subseteq$
\\
$\bigcup\{lin_{\theta_{C,\tt{fifo}}}(w) \mid C \text{ a finite conversation in } \Ka \text{ and } w \in \alpha(\wdspo(\po_C)) \; \}$.     
\qed
\end{corollary}

%%%%%%%%%%%%%%%%%%%%%%%%%%%%%%%%%%%%%%%%%%%%%%%%%%%%%
%%%%%%%%%%%%%%%%%%%%%%%%%%%%%%%%%%%%%%%%%%%%%%%%%%%%%

\section{Conclusion}
\label{Sect-Concl}

%%%%%%%%%%%%%%%%%%%%%%%%%%%%%%%%%%%%%%%%%%%%%%%%%%%%%

%
In Sections~\ref{Sect-MLandPO} and~\ref{Sect-DCA}, we have outlined a method to transform a message log of a distributed, cross-organisational process into a global event log over a distributed communicating alphabet. 
In particular, given message log $\Ka$, action function $\alpha$, and message type function $\tau$, we thus obtain a message log $\alpha(\Ev(\Ka))$ and an $n$-dimensional distributed communicating alphabet $\DAx(\Ka,\alpha,\tau)$.
Since $\alpha(\Ev(\Ka))$ has the prefix property, the synthesis procedure proposed in~\cite{DBLP:journals/topnoc/KwantesK22} and outlined here in Appendix~\ref{AppSInets}, can now be used. 
Given a process discovery algorithm $\mathcal{A}$ for a family $\mathbb{L}$ of languages, the following theorem is an immediate consequence of Theorem 5 from~\cite{DBLP:journals/topnoc/KwantesK22}. 

\begin{theorem}\label{Thm-Syn}  
If $\proj_{\alpha(\K_i)}(\alpha(\Ev(\Ka))) \in \mathbb{L}$, for all $i \in [n]$, 
then $\alpha(\Ev(\Ka)) \subseteq \mathcal{L}(\mathcal{A}_{\Gn}(\alpha(\Ev(\Ka)),\DAx(\Ka,\alpha,\tau))$.
\qed
\end{theorem}

However, the results in~\cite{DBLP:journals/topnoc/KwantesK22} also involve the $\tt{fifo}$-linearisations of the words from global event log $\alpha(\Ev(\Ka))$ because they are always included in the language of the $\Gn$-net $\mathcal{A}_{\Gn}(\alpha(\Ev(\Ka)),\DAx(\Ka,\alpha,\tau))$. 
In general, the linearisations of a word, defined by an assignment function for that word, reflect the causality of its action occurrences induced by the channel policy captured by the given assignment function. 
As $\tt{fifo}$-functions are oblivious to the identities of the messages exchanged, they are only concerned with occurrences of input/output actions. 
This is consistent with the causality and concurrency in the communicating behaviour of \Gn-nets in terms of the occurrences of input and output transitions (actions in $\alpha(\Ev(\Ka))$) meaning that independent occurrences of such transitions will not be ordered and can be swapped. 

Thus, by Corollary~4 in~\cite{DBLP:journals/topnoc/KwantesK22}, we  have the following corollary of Theorem~\ref{Thm-Syn} and Corollary~\ref{Cor-Lin}. 

\begin{corollary}
\label{Cor-Linfifo}
    If $\proj_{\alpha(\K_i)}(\alpha(\Ev(\Ka))) \in \mathbb{L}$, for all $i \in [n]$, then 
    $\alpha(\Ev(\Ka)) \subseteq$ 
    %\lin_{\tt{FIFO}}(\alpha(\Ev(\Ka))) \subseteq
\\    
    $\bigcup\{lin_{\theta_{C,\tt{fifo}}}(w) \mid C \text{ a finite conversation in } \Ka \text{ and } w \in \alpha(\wdspo(\po_C)) \; \}
    \subseteq 
\\    
    \mathcal{L}(\mathcal{A}_{\Gn}(\alpha(\Ev(\Ka)),\DAx(\Ka,\alpha,\tau))$.
\qed
    \end{corollary}
    
Finally, the inclusions in this corollary are equalities if two conditions are satisfied: 
(1) the $\tt{fifo}$-function is message preserving and (2) $\alpha(\Ev(\Ka))$ is complete with respect to the discovered \Gn-net
(see Appendix~\ref{AppSInets} and Definition~9 in~\cite{DBLP:journals/topnoc/KwantesK22}).   
    
\begin{theorem}\label{T07mta}
If $\proj_{\alpha(\K_i)}(\alpha(\Ev(\Ka))) \in \mathbb{L}$, for all $i \in [n]$, then
$\alpha(\Ev(\Ka)) = \mathcal{L}(\mathcal{A}_{\Gn}(\alpha(\LK_{\Ka}),\DAx(\Ka,\alpha)))$, provided 
$\theta_C = \theta_{C,\tt{fifo}}$ for all finite conversations $C$ in $\Ka$
and 
$\alpha(\Ev(\Ka)))$ is complete with respect to $\mathcal{A}_{\Gn}(\alpha(\Ev(\Ka)),\DAx(\Ka,\alpha,\tau))$.
\qed
\end{theorem}

%%%%%%%%%%%%%%%%%%%%%%%%%%%%%%%%%%%%%%%%%%%%%%%%%%%%%
%%%%%%%%%%%%%%%%%%%%%%%%%%%%%%%%%%%%%%%%%%%%%%%%%%%%%

\section{Discussion}
\label{Sect-Disc}

%%%%%%%%%%%%%%%%%%%%%%%%%%%%%%%%%%%%%%%%%%%%%%%%%%%%%

%
The main aim of this paper was to show how the communicating behaviour of a cross-organisational process in the form of a global event log over a distributed communicating alphabet, can be inferred from a message log, \ie observations of message exchanges between component business processes.
This should be done in such a way that the derived behaviour faithfully represents the relations between the communications as recorded in the message log. 
This made it possible to apply the results from~\cite{DBLP:journals/topnoc/KwantesK22}. 
So, existing algorithms for the discovery of Petri net models of local business processes that can be leveraged to the discovery of Petri net models of cross-organisational processes from global event logs, can now also be used to discover Industry nets from message logs. 

In~\cite{DBLP:journals/topnoc/KwantesK22} and in this paper as well, a global event log is understood as a prefix-closed, finitary language consisting of sequences of executed actions. 
However, being based on message logs, the sequences in the global event logs in this paper have no occurrences of internal actions. 
Since global event logs are in general infinite finitary languages, message logs are not a priori required to be finite. 
As a consequence, also conversations may be infinite sets. 
Such conversations define infinitary languages, which by Theorem~\ref{Thm-limConv} are fully described by the limits of their finite prefixes. 
In fact, by Lemma~\ref{Lem-prefseq=seqpref}, the words defined by the finite subconversations of any infinite conversation are exactly the finite prefixes of its infinite sequences. 
Hence, our focus was on finite conversations and the prefix-closed, finitary languages defined by them. 
Next, message events had to be interpreted as executions of (output or input) actions of local processes, reflecting the idea that a single action may communicate many messages during a run (conversation) of the system. 
Moreover, our system models (Enterprise nets and Industry nets) are finitely specified and so are the distributed communicating alphabets associated with the global event logs.  
This led to the concept of an action function describing how to consistently label message events with finitely many different action names. 
The types of the message events had to be lifted to the level of action names in the distributed alphabet in such a way that the action labels of complementary message events have the same type. 
In other words, complementary actions should have the same type exactly what is required by the definition of a distributed communicating alphabet. 
Finally, note that -- whereas the process function of the message log is explicitly represented in the distributed communicating alphabet -- the case and time function are used to define its event language via the partial orders defined by the conversations.

The results from~\cite{DBLP:journals/topnoc/KwantesK22} show how the global event log derived from a message log can be used to discover an Industry net the language of which includes the global event log (Theorem~\ref{Thm-Syn}). 
Following~\cite{DBLP:journals/topnoc/KwantesK22}, we then turned to assignment functions to further investigate this inclusion.  
Assignment functions relate occurrences of input and output actions in a sequence and thus induce a partial order on these occurrences. 
They can be seen as channel policies describing the order in which sent messages were received.
In case of global event logs derived from message logs, message preserving assignment functions can be defined on basis of the underlying message events. 
Only these message preserving assignment functions (shared by all sequences that belong to the same conversation) preserve the order of sending at the receiving end. 
By Theorem~\ref{Thm-IsoPO}, the partial order on the action occurrences defined by a message preserving assignment function associated with a conversation, coincides with the partial order of that conversation. 
However, in general, global event logs are more abstract and know only occurrences of actions. 
As a consequence, the causality between input and output actions cannot be based on the identities of the messages exchanged. 
In its least restrictive form, this causality is expressed through the prefix property as captured by $\tt{fifo}$-functions. 
This explains why the linearisations defined by message preserving assignment functions are included in the linearisations defined by first-in-first-out functions (\cf Corollary~\ref{Cor-Linfifo}). 

\medskip
To conclude, as illustrated in this paper, the translation of communication (message exchanges) into global event logs makes it possible to extend established local process mining methods and techniques to the level of collaborating organisations. 
We consider this as a useful step towards overcoming this challenge from~\cite{DBLP:conf/bpm/AalstAM11}. 
It would be interesting to investigate if this technique based on channel policies and (labelled) partial orders could also be useful in other settings.

%%%%%%%%%%%%%%%%%%%%%%%%%%%%%

\subsection*{Acknowledgments} 

%%%%%%%%%%%%%%%%%%%%%%%%%%%%%

We are grateful to Hendrik Jan Hoogeboom for his comments on a preliminary version of this paper and for his invaluable assistance with the figures.

\bibliographystyle{incld/splncs03}
\bibliography{incld/Archive.bib}

\begin{appendix}

\section*{Appendix}

\section{Proof of Lemma~\ref{Lem-seqlimpref}} 
\label{AppProofs} 

\textbf{Lemma~\ref{Lem-seqlimpref}}\;  Let $\po = (A,\leq_A)$ be a well-founded, partial order with $A$ an infinite set. 
    Then $\wdspo(\po) = \llim(\Pref(\wdspo(\po)))$.  
\begin{proof}
The inclusion of $\wdspo(\po)$ in $\llim(\Pref(\wdspo(\po)))$ follows from our earlier observation that $L \subseteq \lim(\pref(L))$ for all infinitary languages $L$. 

To prove the converse, consider a $w \in \llim(\Pref(\wdspo(\po))$. 
Since $\po$ is well-founded and $A$ is infinite, such $w \in A^\omega$ exists. 
Moreover, $w[i] \in \Pref(\wdspo(\po))$ for infinitely many $i$. 
This implies that $w[i] \in \Pref(\wdspo(\po))$ for all $i \in {\mathbb N}$. 
\\
Let $D = \{d \in A \mid d \not\in\ab(w)\}$ comprise all elements of $A$ that do not occur in $w$. 
Let $\po_D = (D,\leq_D)$ be the sub-poset of $\po$ induced by $D$ and let $(D,\preceq_D) \in \lex (D,\leq_D)$ be a linear extension of $\po_D$. 
We set $\preceq \; = \; 
\leq_A \; \cup \; 
\preceq_D \cup \;
\{(w(i),w(j)) \mid 1 \leq i \leq j\} \cup 
\{(w(i),d) \mid i \geq 1 \text{ and } d \in D\}$. 
\\
We first argue that there exist no $a \in D$ and $b \in A \setminus D$ with $a\leq_A b$. 
\hfill 
(*)
\\
This can be seen as follows.
Assume to the contrary that $a \in D$ and $b \in A \setminus D$ are such that $a\leq_A b$. 
Thus $b = w(j)$ for some $j \geq 1$. 
Since $w[j] \in \Pref(\wdspo(\po))$, it follows that $\dc_\po(w(j)) = \{c \in A \mid c \leq_A w(j)\} \subseteq \{w(1), \ldots w(j)\}$ (see Lemma~\ref{rem-pastlo}) which implies that $a = w(i)$
for some $i \leq j$, in contradiction with $a \in D$. 
Thus such $a,b$ do not exist.
\\
Next we prove that $(A,\preceq) \in \lex(\po)$. 
\\
By definition, $\leq_A \; \subseteq \; \preceq$.
In addition, clearly, $a \preceq b$ or $b \preceq a$ for all $a,b \in A$. 
Hence, what remains to be shown is that $\preceq$ is a partial order on $A$. 
\\
Since $\leq_A \; \subseteq \;  \preceq$, we immediately have that $\preceq$ is reflexive.
\\
To prove that $\preceq$ is antisymetric, let $a,b \in A$ be such that $a \preceq b$ and $b \preceq a$. 
Let $a,b \in D$.
Then, as $\preceq_D$ is total, we may assume without loss of generality that $a \preceq_D b$.
There are two possible cases.
If $b \preceq_D a$, then $a = b$ by the antisymmetry of $\preceq_D$. 
And if $b \leq_A a$, then also $b \preceq_D a$ since $\preceq_D$ is a linear extension of the sub-poset of $(A,\leq_A)$ induced by $D$.
Thus, again $a = b$ by the antisymmetry of $\preceq_D$. 
If both $a,b \in A\setminus D$, then $a = w(i)$ and $b = w(j)$ for some $i,j \geq 1$. 
Assume without loss of generality that $i \leq j$. 
Since $w[j]\in \Pref(\wdspo(\po))$, the pair $(w(i),w(j))$ thus belongs to at least one linear extension of $\leq_A$.
On the other hand, if $w_j \leq_A w_i$ the pair $(w(j),w(i))$ belongs to all linear extensions of $\leq_A$. 
It follows that $i = j$ and thus $a = b $. 
Finally, if $a \in D$ and $b \in A \setminus D$, then $a\preceq b$ implies $a \leq_A b$ in contradiction with (*) above. 
So we are done.
\\
To prove the transitivity of $\preceq$, consider $a,b,c \in A$ such that $a \preceq b$ and $b\preceq c$. 
If $a \in A\setminus D$ and $c \in D$, then $a \preceq c$ by definition.
Moreover, by (*) above, the only remaining
relevant cases are that $a,b,c \in A\setminus D$ and either (1) $(a,b) \in \; \leq_A$ and $(b,c) \not\in \; \leq_A$ 
or, conversely, (2) 
$(a,b) \not\in \; \leq_A$ and
$(b,c) \in \; \leq_A$.
In all other situations, transitivity follows from the transitivity of the four relations constituting $\preceq$ and their combinations.
Now, let $i,j,k \geq 1$ be such that $w(i) = a$, $w(j) = b$, and $w(k) = c$. 
$w_j \in \dc(w(k)) \subseteq \{w(1), \ldots , w(k)\}$. 
If (1) holds, then $w(i) \leq_A w(j)$ with $w[j] \in \Pref(\wdspo(\po))$.
Hence, by Lemma~\ref{rem-pastlo}, $w(i) \in \dc(w(j)) \subseteq \{w(1), \ldots , w(j)\}$ and so $i \leq j$. 
Moreover, $w(j) = b \preceq c = w(k)$ and so $j \leq k$.
Hence $i \leq j \leq k$ and $a = w(i) \preceq w(k) = c$ follows.
If (2) holds, a similar reasoning can be used to argue that $a \preceq c$. 
\\
We conclude by showing that $w = \sigma_\lo$ where $\lo = (A,\preceq)$.
From (*) and the definition of $\preceq$, we have that $w(1)$ is the unique minimal element of $\lo$. 
Now, consider any pair $w(i),w(i+1)$ where $i \geq 1$, and assume that $w(i) \preceq a \preceq w(i+1)$ for some $a \in A$. 
From (*) it follows that $a \not\in D$. 
So there exists $j \geq 1$ such that $a = w(j)$. 
Since $\preceq$ is a linear order, it must be that $j = i$ or $j = i+1$ and we are done. 
%
%Now consider the prefix $w[i+1]$ of $w$. 
%
%Since $w[i+1] \in \pref(\wdspo(\po))$, we know that $w(i) \leq_A a \leq_A w(i+1)$ implies that $a = w(i)$ or $a = w(i+1)$. 
%
Hence, for all $i\in\dom(w)$, $w(i)$ is the direct predecessor of $w(i+1)$ in $\lo$. 
We conclude that $w = \sigma_\lo$ and thus $w \in \wdspo(\po)$. 
\qed
\end{proof}

\section{Petri nets} 
\label{AppPnets} 

A Petri net is a triple $N=(P,T,F)$, where 
$P$ is a finite set of \emph{places},
$T$ is a finite, non-empty set\footnote{$T$ is an alphabet.} of \emph{transitions}, 
$P \cap T=\emptyset$, 
and $F \subseteq (P \times T) \cup (T \times P)$ is a set of \emph{arcs}.

Let $N=(P,T,F)$ be a Petri net.
A \emph{marking of} $N$ is a function $\mu: P \rightarrow \mathbb{N}$ that assigns $\mu(p)$  \emph{tokens} to each place $p$ of $N$. 
Let $t \in T$ and let $\mu$ be a marking of $N$. 
Then $t$ is \emph{enabled at} $\mu$ if $\mu(p) > 0$ 
for all $p \in P$ such that $(p,t) \in F$. 
If $t$ is enabled at $\mu$, it may \emph{occur}, and thus lead to a new marking $\mu'$ of $N$. 
Here $\mu'$ is defined by 
$\mu'(p) = \mu(p)-1$ if $(p,t) \in F$ and $(t,p) \not\in F$; 
$\mu'(p) = \mu(p)+1$ if $(p,t) \not\in F$ and $(t,p) \in F$; and 
$\mu'(p) = \mu(p)$ otherwise. 
We write $\mu \xrightarrow{t}_N \mu'$ to denote that the occurrence of $t$ at $\mu$ in $N$ leads to $\mu'$. 
Furthermore, $\mu \xrightarrow{\lambda}_N \mu$ for every marking $\mu$ of $N$; 
and in case $v \in T^*$ with $v = wt$ for some $w \in T^*$ and some $t \in T$, then 
$\mu \xrightarrow{v}_N \mu'$ if there exists a marking $\mu''$ of $N$ such that $\mu \xrightarrow{w}_N \mu''$ and $\mu'' \xrightarrow{t}_N \mu'$.
A place $p \in P$ is  a \emph{source place} of $N$ if 
there is no $t \in T$ such that $(t,p) \in F$. 
The marking $\mu_0$ of $N$ such that, for all $p \in P$, $\mu_0(p)=1$ if $p$ is a source place and $\mu_0(p)=0$ otherwise, is the \emph{default initial marking} of $N$.
The \emph{language of} $N$ is 
$\mathcal{L}(N)$ = 
$\{w \in T^* \mid \exists \mu \mbox{ such that }\mu_0 \xrightarrow{w}_N \mu  \}$. 
Note that $\mathcal{L}(N)$ is a finitary, prefix-closed language.

\section{Enterprise nets and Industry nets} 
\label{AppInets} 

\emph{Enterprise nets} (introduced in~\cite{DBLP:conf/apn/KwantesK18}) are Petri nets with typed output and input transitions that can be used for asynchronous communication with other Enterprise nets.   

\begin{definition} 
An \emph{Enterprise net} (or \Ln-net, for short) is a tuple $\LN = (P,$ $\langle T_{int},T_{inp},T_{out} \rangle,$ $ F,$ $M)$ 
such that
$T_{int},T_{inp}$, and $T_{out}$ are pairwise disjoint sets;
$T_{int}$ is the set of \emph{internal transitions} of $\LN$,
$T_{inp}$ its set of \emph{input transitions},
and $T_{out}$ its set of \emph{output transitions};
furthermore, the \emph{underlying Petri net of $\LN$}, $und(\LN) = (P,T_{int} \cup T_{inp} \cup T_{out},F)$,
is a Petri net with exactly one source place;
finally $M$ is a function $M: T_{inp} \cup T_{out} \rightarrow \mathcal T$ with $\mathcal T$ the set of \emph{message types} of $\LN$.
\qed
\end{definition}

\noindent
$\mathcal{L}(und(\LN))$,  also denoted as $\mathcal{L}(\LN)$, is the \emph{language of} the Enterprise net $\LN$. 

\medskip
Communications take place between pairs of Enterprise nets via (initially unmarked) intermediate places between an output transition of one Enterprise net and an input transition (of the same type) of another Enterprise net. 
When the output transition occurs, it puts a token in the intermediate place and every occurrence of the input transition takes a token from that intermediate place. 
Thus, to compose enterprise nets, the output and input transitions of the component nets have to be matched. 

Let $\V=\{\LN_i \mid i \in [n] \}$ with $\LN_i=(P_i, \langle T_{i,int},T_{i,inp},T_{i,out} \rangle,F_i,M_i)$ for each $i \in [n]$, 
be a set of $n \geq 1$ \Ln-nets that have no shared elements (places and transitions).
A bijection $\varphi: \bigcup_{i\in [n]} T_{i,out} \rightarrow \bigcup_{j \in [n]} T_{j,inp}$ is a \emph{matching over} $\V$ if, whenever $t \in T_{i,out}$ and $\varphi(t)\in T_{j,inp}$, for some $i,j$, then $i\neq j$ and $M_i(t)=M_j(\varphi(t))$.
Note that it may be that no matching over $V$ exists. 
If a matching exist, $\V$ is said to be \emph{composable}.
In that case, an \emph{Industry net} can be composed from the Enterprise nets from $\V$ by connecting matching output and input transitions via new places. 

\begin{definition}
\label{Def-Inet}
Let $n \geq 2$. Let $\V=\{\LN_i: i \in [n] \}$ be a composable set of \Ln-nets
with $\LN_i=(P_i,\langle T_{i,int},T_{i,inp},T_{i,out}\rangle,F_i,M_i)$ and $T_i= T_{i,int} \cup T_{i,inp} \cup  T_{i,out}$ for all $i \in [n]$.
Let $\varphi$ be a matching over $\V$.
\\
Then $P(\V,\varphi)=\{[t,\varphi(t)] \mid t \in T_{i,out}, i \in [n]\}$ is the set of \emph{channel places} of $V$ and $\varphi$, and $F(\V,\varphi)=$ $\{(t,[t,\varphi(t)]) \mid $ $ t \in T_{i,out}, i \in [n] \} \cup \{([t,\varphi(t)],\varphi(t)) \mid $ $ t \in T_{i,out}, i \in [n] \}$ is the set of \emph{channel arcs} of $V$ and $\varphi$.
The sets $P(\V,\varphi)$, $F(\V,\varphi)$, and $P_i$, $T_i$, $F_i$, where $i \in [n]$, are all pairwise disjoint.
\\
The \emph{Industry net (or \Gn-net, for short) over $(\V,\varphi)$} is the Petri net $\GN(\V,\varphi) = (P,T,F)$ with
$P=\bigcup_{i \in [n]} P_i \cup P(\V,\varphi)$,
$T=\bigcup_{i \in [n]} T_i$, and
$F=\bigcup_{i \in [n]} F_i \cup F(\V,\varphi)$.
\qed
\end{definition}

Since channel places are not source places and consequently not marked by
the default initial marking, it follows that in every word $w$ from the language of an Industry net $\GN(\V,\varphi)$, the number of occurrences of an input transition is at most the number of occurrences of its corresponding output transition. 
Moreover, this also hold for all prefixes of $w$: 
each occurrence of an input transition $t$ must be preceded by a unique occurrence of the output transition $\varphi^{-1}(t)$. 
In other words, $\mathcal{L}(\GN(\V,\varphi))$ has the prefix property with respect to $\varphi$.

\section{Synthesising Industry nets from event logs} 
\label{AppSInets} 

In~\cite{DBLP:journals/topnoc/KwantesK22}, process discovery algorithms are described in the following way. 

\begin{definition}\label{Def-A}  
Let $\mathbb{L}$ be a family of languages.
A \emph{process discovery algorithm} $\mathcal{A}$ for $\mathbb{L}$
is an algorithm that computes for all $L \in \mathbb{L}$,
a Petri net $\mathcal{A} (L)=(P,T,F)$ with a single source place
 such that $T=\Ab(L)$ and $L \subseteq \mathcal{L}(\mathcal{A}(L))$.
\qed
\end{definition}

Note that since Petri nets have a set of transitions that is an alphabet (a non-empty and finite set), the above implies that any family of languages for which a process discovery algorithm exists, comprises only languages that can be defined over a finite alphabet. 
Furthermore, to facilitate their conversion to Enterprise nets, the Petri nets discovered are required here to have a single source place, as do Enterprise nets.

To leverage a process discovery algorithm for the discovery of Petri net models for individual component processes, to discovering a distributed process model for communicating processes, additional information is needed. 
In \cite{DBLP:journals/topnoc/KwantesK22}, next to the input of the global event log (in the form of a language), also the number of participating individual processes (leading to Enterprise nets) is given together with their communication channels described by pairs of matching input and output actions. 
All this is captured in the form of an alphabet together with a specification of the role of the actions (output, input, internal) and their distribution over the component processes. 

\begin{definition}
\label{Def-DCA}
Let $n \geq 1$. 
An ($n$\emph{-dimensional})  \emph{distributed communicating alphabet} ($n$-\DA{} or $\DA$, for short) is a tuple
$$\DAx=([\Sigma_1, \ldots, \Sigma_n], [ \Sigma_{int},\Sigma_{inp},\Sigma_{out}] ,mt, cp )$$ 
such that 
\\
(1) $\Sigma_1, \ldots, \Sigma_n$ are pairwise disjoint alphabets;
\\
(2) $\Sigma_{int}$, $\Sigma_{inp}$, $\Sigma_{out}$ are pairwise disjoint, finite sets
of \emph{internal actions}, \emph{input actions}, and \emph{output actions}, respectively;
\\
(3) 
%$\Sigma_\DAx = $
$\bigcup_{i \in [n]} \Sigma_i = \Sigma_{int} \cup \Sigma_{inp} \cup \Sigma_{out}$; 
%is the \emph{underlying alphabet of} $\DAx$;
\\
(4) $mt: \Sigma_{inp} \cup \Sigma_{out} \rightarrow \mathcal{T}$ where $\mathcal T$ is a set of \emph{message types}; 
and 
\\
(5) 
$cp : \Sigma_{inp} \cup \Sigma_{out} \rightarrow \Sigma_{inp} \cup \Sigma_{out}$; if $n = 1$, then $cp$ is not defined 
\footnote{If $n=1$, then $cp = \emptyset$ and usually omitted from the specification of $\DAx$.}; otherwise $cp$ is a complement function with 
\\
\phantom{space}
(5.1) $cp(\Sigma_{inp}) = \Sigma_{out}$ and 
$cp(\Sigma_{out}) = \Sigma_{inp}$; 
\\
\phantom{space}
(5.2) $cp(\Sigma_i) \cap \Sigma_i = \emptyset$ for all $i \in [n]$; 
and 
\\
\phantom{space}
(5.3) 
$mt(cp(a))=mt(a)$ for all $a \in \Sigma_{inp} \cup \Sigma_{out}$.
\qed 
\end{definition}

This definition reflects that local processes have disjoint sets of actions and that each action is either an internal, an input, or an output action of the process it belongs to. 
Moreover, with each input action corresponds a unique output action of another process and there is exactly one such corresponding input action for every output action. 
Finally, all input and output actions have a message type and complementary input/output actions have the same type. 
Actually, the distribution of transitions of an \Gn-net $\GN(\V,\varphi)=(P,T,F)$  over its component \Ln-nets in $\V$ with their role as internal, input or output actions and the message types from the \Ln-nets, forms a \DA{} with underlying alphabet $T$ and complement function $\varphi \cup \varphi^{-1}$.  

In~\cite{DBLP:journals/topnoc/KwantesK22}, it is demonstrated how from a process discovery algorithm $\mathcal{A}$ for language family $\mathbb{L}$, an algorithm can be constructed to synthesise an Industry net from an event log and a distributed communicating alphabet. 

\begin{definition}
\label{Def-discA}
(1) 
The \emph{\Ln-net discovery algorithm derived from $\mathcal{A}$}, is the algorithm $\mathcal{A}_{\Ln}$ that computes, for any pair $(L,\DAx)$
such that
$L \in \mathbb{L}$, $\mathcal{A}(L) = (P,T,F)$, and 
$\DAx = (\langle T \rangle,\langle T_{int},T_{inp},T_{out} \rangle ,mt)$ is a 1-\DA{},  
% with $T = \Ab(L)$ as its underlying alphabet, DAT IS NIET OK, T is al gegeven in het net.
%
the \Ln-net $\mathcal{A}_{\Ln}(L,\DAx) =  (P,\langle T_{int}, $ $T_{inp}, T_{out} \rangle,F,mt)$.

(2) The \emph{\Gn-net discovery algorithm derived from $\mathcal{A}$}, is the algorithm  $\mathcal{A}_{\Gn}$
that computes, for any pair $(L,\DAx)$
such that $\DAx$ is an $n$-\DA{} where $n \geq 2$, 
$\DAx =(\langle \Sigma_1, \ldots, \Sigma_n \rangle,  \langle \Sigma_{int},\Sigma_{inp},\Sigma_{out} \rangle ,mt, cp )$, $\Ab(L)= \bigcup_{i \in [n]} \Sigma_i$, and $\proj_{\Sigma_i}(L) \in \mathbb{L}$ for all $i \in [n]$, 
the \Gn-net $\mathcal{A}_{\Gn}(L,\DAx)=\GN(\V,\varphi)$ with
$\V=\{\mathcal{A}_{\Ln}(\proj_{\Sigma_i}(L),\DAx_i) \mid i \in [n] \}$ and $\varphi(a)=cp(a)$ for all $a \in \Sigma_{out}$.
\qed
\end{definition}

Assume a process discovery algorithm $\mathcal{A}$ for $\mathbb{L}$
as given. 
Let $n \geq 2$ and let $(L,\DAx)$ be such that 
$\DAx =(\langle \Sigma_1, \ldots, \Sigma_n \rangle,  \langle \Sigma_{int},\Sigma_{inp},\Sigma_{out} \rangle ,mt, cp )$, $\Ab(L)= \bigcup_{i \in [n]} \Sigma_i$, and $\proj_{\Sigma_i}(L) \in \mathbb{L}$ for all $i \in [n]$. 
By Definition~\ref{Def-A} and Definition~\ref{Def-discA}, 
$\proj_{\Sigma_i}(L) \subseteq \mathcal{L}(\mathcal{A}_{\Ln}(\proj_{\Sigma_i}(L),\DAx_i)) = \mathcal{L}(\mathcal{A}_{\Ln}(\proj_{\Sigma_i}(L))$. 
Furthermore, by Theorem 5 from~\cite{DBLP:journals/topnoc/KwantesK22}, we have $L \subseteq \mathcal{L}(\mathcal{A}_{\Gn}(L,\DAx))$ in case $L$ has the prefix property with respect to $\varphi$.  
In~\cite{DBLP:journals/topnoc/KwantesK22}, it is also shown how this result can be strengthened to an equality when $L$ and $\mathcal{A}_{\Gn}(L,\DAx))$ satisfy the additional property that every word from the language of $\mathcal{A}_{\Gn}(L,\DAx))$ is represented in $L$ by a word $w' \in L$, \ie $\proj_{\Sigma_i}(w) = \proj_{\Sigma_i}(w')$ for all $i \in [n]$ - in which case  $L$ is said to be \emph{complete} with respect to $\mathcal{A}_{\Gn}(L,\DAx)$. 

\end{appendix}

\end{document}